\documentclass[12pt,draftclsnofoot,onecolumn]{IEEEtran}
\IEEEoverridecommandlockouts
\usepackage{graphicx}
\usepackage{amsmath}
\usepackage{amsfonts}
\usepackage{amssymb}
\usepackage{cite}
\usepackage{color}
\usepackage{multirow}
\usepackage{enumerate}
\usepackage{tabularx}
\usepackage{amsthm}
\usepackage{booktabs}
\usepackage[figurename=Fig., tablename=Table]{caption}

\usepackage{algorithmicx}
\usepackage{algorithm}
\usepackage{multicol}
\usepackage{algpseudocode}

\newtheorem{theorem}{Theorem}
\newtheorem{define}{Definition}
\newtheorem{remark}{Remark}
\newcolumntype{L}[1]{>{\raggedright\arraybackslash}p{#1}}
\newcolumntype{C}[1]{>{\centering\arraybackslash}p{#1}}
\newcolumntype{R}[1]{>{\raggedleft\arraybackslash}p{#1}}

\usepackage{times, epsfig}
\usepackage{subfig}

\newcommand{\E}{\mathbb{E}}

\makeatletter
\newcommand{\multiline}[1]{%
	\begin{tabularx}{\dimexpr\linewidth-\ALG@thistlm}[t]{@{}X@{}}
		#1
	\end{tabularx}
}
\makeatother

\newtheorem{Lem}{Lemma}

\newlength{\figwidth}
\begin{document}
	
	\setlength{\pdfpagewidth}{8.5in}
	\setlength{\pdfpageheight}{11in}
\title{\LARGE Power Allocation for Compute-and-Forward over Fading Channels}

\author{Lanwei Zhang,~\IEEEmembership{Student Member,~IEEE,}
        Jamie Evans,~\IEEEmembership{Senior Member,~IEEE,}
        and~Jingge Zhu,~\IEEEmembership{Member,~IEEE.}
}

\maketitle

\begin{abstract}
Compute-and-forward (CF) is a relaying strategy which allows the relay to decode a linear combination of the transmitted messages. This work studies the optimal power allocation problem for the CF scheme in fast fading channels for maximizing the symmetric computation rate, which is a non-convex optimization problem with no simple analytical or numerical solutions. In the first part of the paper, we investigate the problem when there are finitely many channel states (discrete case). We establish several important properties of the optimal solutions and show that if all users share the same power allocation policy (symmetric policy), the optimal solution takes the form of a water-filling type when the power constraint exceeds a certain threshold. However, if asymmetric policies are allowed, the optimal solution does not take this form for any power constraint. We propose a low-complexity order-based algorithm for both scenarios and compare its performance with baseline algorithms. In the second part of the paper, we state relevant results when the channel coefficients are modelled as continuous random variables (continuous case) and propose a similar low-complexity iterative algorithm for the symmetric policy scenario. Numerical results are provided for both discrete and continuous cases. It is shown that in general our proposed algorithm finds good suboptimal solutions with low complexity, and for some examples considered, finds an exact optimal solution.
\end{abstract}


\section{Introduction}

In a wireless network, the signal sent by a transmitter is not only received by the intended receiver but also by surrounding nodes due to the broadcast nature of the medium. In the same way, the received signal at a node will include both the desired message and the unwanted interference signal. Traditionally, orthogonal transmission schemes are applied to avoid interference so that the receiver only receive the transmitted signal from the desired transmitter. However, these methods will suffer from a diminishing rate, especially when many nodes want to send messages simultaneously \cite{Tse05fundamentalsof}.

If the transmitters and the receivers cannot communicate with each other directly, some intermediate nodes can be used as relays to help their communication. Conventional relay strategies include amplify-and-forward and decode-and-forward where either a simple analog operation is used, or the messages from different transmitters are treated "separately" in a decoding process. The discovery of network coding\cite{Ahlswede00networkcoding} shows that the traditional relay strategies are not sufficient to achieve optimal throughput. To address the same problem over noisy channels, physical-layer network coding (PNC) was proposed in \cite{Zhang06PLNC}. It is a non-orthogonal approach which exploits the physical nature of the electromagnetic waves. The relays will recover a set of functions of the signals from multiple transmitters and then learn the desired messages from the functions. 

Compute-and-forward (CF) is a linear PNC scheme proposed in \cite{Nazer11CF}. It allows the relay to decode an integer linear combination of the messages from multiple transmitters, and either forward it to the next receiver or combine it with other linear combinations to recover individual messages. To achieve this, nested lattice codes are used to guarantee that any integer linear combination of the codewords is still a codeword. Due to the favourable algebraic property of lattice codes, the idea of CF has been applied to different AWGN network models and results in good performances that cannot be achieved otherwise\cite{Nazer11CF,Zhu17,Nazer16,Hong13,Ordentlich17,Zhan14,He18}. Due to the difference between the channel coefficients and the integer coefficients of the linear combinations, the achievable computation rate to decode the linear combination will have a loss. One can allocate different transmission powers to different transmitters in order to match the channel coefficients to the integers if the channel state information (CSI) is known to the transmitters. When considering fading channels, where the channel coefficients are modelled as random variables (RV), one should further decide judiciously whether the transmission should happen and how much power should be allocated for different channel states to maximize the average transmission rate. However, since the rate expression of the CF scheme is non-smooth and non-convex, the resulting optimization problem is thus non-convex. In this paper, we will explore the characteristics of the optimal solutions to this non-convex problem and discuss several low-complexity algorithms.

\subsection{Related work}
There is a rich collection of work on physical-layer network coding and the compute-and-forward scheme in many network scenarios. It has been shown in \cite{Zhang06PLNC} that for a two-user two-way relay network, the PNC scheme has nearly twice the throughput of the orthogonal scheme. It is achieved by decoding the sum of the transmitted signals at the relay and sending out the sum. Since both users know their own messages, they can learn the message of the other user using the sum. 

By using the CF scheme, the relay will recover the integer linear combination of the transmitted signals. After collecting enough independent combinations, the receiver can recover all the individual messages. Following this idea, the integer forcing linear receiver has been proposed in \cite{Zhan14}. This linear receiver has low complexity and significantly outperforms conventional linear architectures such as zero forcing. It can also apply to the multiple-input multiple-output (MIMO) channel with no coding across transmit antennas while still achieving the optimal diversity-multiplexing tradeoff. The original CF scheme is expanded in \cite{Nazer16,Zhu17}, which allows unequal rates among the transmitters. In both papers, the authors show that their proposed low-complexity scheme can achieve the channel capacity for the Gaussian multiple-access channels (MAC) given high signal-to-noise ratio (SNR). Apart from the Gaussian MAC, CF is also applied to the random access Gaussian channel with a relatively large number of active users\cite{Ordentlich17}. It is shown that the energy-per-bit required by each user is significantly smaller than the popular solutions such as treating interference as noise. Moreover, the reverse-CF scheme has been proposed for the Gaussian broadcast channels as the duality of the MAC\cite{Hong13,He18}. In this scheme, a precoded process is required at the base station. In \cite{zhang24}, a CF-based joint coding scheme for Gaussian fast fading multiple access channels with CSIR only has been proposed and its ergodic capacity achievability has been discussed. 

To fully exploit the advantage of the CF scheme, the coefficients of the decoded linear combinations should be determined carefully. For example, to recover the messages with the same amount of linear combinations as the amount of users, we should decide all the combinations are linearly independent. The integer coefficients should also be adapted to the channel state. Optimizing these coefficients is a Shortest Lattice Vector (SLV) problem. In \cite{Wei12}, the authors design the coefficients by Fincke-Pohst-based candidate set searching algorithm and network coding system matrix constructing algorithm, with which the transmission rate of the multi-source multi-relay system is maximized. It has also been shown in \cite{Sahraei14} that this SLV problem from the CF scheme can be solved in low polynomial complexity and an explicit deterministic algorithm that is guaranteed to find the optimal solution has been proposed. 

Another aspect of this CF scheme is to understand the power allocation among the different users in the network. As discussed in the previous subsection, if CSI is available at the transmitters, allocating the transmitted powers can match the channel state with the integer coefficients, which will give a better rate. To this end, the authors in \cite{Tran14} propose an iterative method to optimize the achievable rate over the integer coefficients and the power allocation alternatively in a multi-user multi-relay network with the CF scheme. However, in the power allocation phase, the Lagrange method is applied, which is a sub-optimal approach for the non-convex problem. We will see the gap between the optimal approach and this approach in the later sections. Another network model is considered in\cite{Soussi14}, which includes two users, a single relay and one receiver. The relay helps the transmission from the users to the receiver while there are direct links between both users and the receiver. The CF scheme is applied at both the relay and the receiver. Then the receiver will learn a linear combination directly from both users and receive another linear combination from the relay. The power allocation problem over the users is then addressed. The authors formulate the problem by approximated geometric programming, which is sub-optimal due to the approximation. Another power allocation problem involving CF discussed in \cite{Chen15} applies a time division protocol in a relay network, where only the first phase applies the CF scheme. The optimization problem focuses on the time fraction and the power allocation among the relays. 

Among all the CF power allocation papers, the optimal solutions are not discussed in detail. More importantly, none of the existing papers studies the CF model with multiple transmitters and one receiver in a fading scenario, which is the basic building block in any CF-based schemes. Under the traditional orthogonal scheme, the classical water-filling algorithm provides an optimal power allocation policy over fading channels \cite{Goldsmith97}. However, as we will show later, more complex power allocation policies are required to optimize the computation rate under the CF scheme. In particular, the optimization problem is in general a non-convex problem with a non-smooth objective function. As a result, most off-the-shelf optimization algorithms (for example, \texttt{fmincon} in MATLAB) are strictly sub-optimal. The problem does not admit a simple optimal solution either numerically or analytically, which might explain why there is no prior work on this particular problem although the CF scheme has received a lot of attention in the field. 

\subsection{Paper contributions}

In this paper, we consider a multiple access fading channel, where the receiver requires a linear combination of the messages with predetermined equation coefficients. To the best of our knowledge, this is the first paper that discusses the power allocation problem for CF over both users and fading channels. We formulate the power allocation problem under the original CF scheme where users are equipped with lattice codes of the same rate. In general, the problem is nontrivial to solve due to its non-convexity. The main contributions of this paper are as follows.
\begin{itemize}
    \item For the case with finitely many channel states, we establish several important properties of the optimal solutions of this non-convex problem. 
    \item A simplified power allocation scheme where all users share the same allocation policy (called symmetric policy) is studied. It is shown that when the power constraint exceeds a certain threshold, the optimal symmetric policy takes the form of a water-filling type solution. We propose a low complexity order-based algorithm to calculate the power allocation scheme for all values of power constraints.
    \item We extend the proposed algorithm to calculate a general power allocation policy where different users can have different power allocation policies. 
    \item We state relevant results for the problem when the channel coefficients are modelled as continuous random variables (continuous case). We propose a similar low-complexity iterative algorithm for the symmetric policy for the continuous case.
    \item Numerical results are provided for both discrete case and continuous case. The performance of the proposed algorithms is compared with several baseline algorithms, including an off-the-shelf algorithm (\texttt{fmincon} in MATLAB). In general, our proposed algorithm finds good suboptimal solutions with low complexity, and for some examples, it finds an exact optimal solution.
\end{itemize}

\section{Problem formulation}\label{SM}

We consider an $L$-user Gaussian multiple access channel (MAC) with real-valued channel coefficients $h_1,\ldots,h_L$ and the additive noise at the receiver is assumed to have zero mean and unit variance. 

\subsection{CF computation rate}
Let $\textbf{w}_1,\ldots,\textbf{w}_L$ drawn from the vector space $\mathbb{F}_p^k$ represent the messages from $L$ users, where $p$ is a primal number. The encoded signals to be transmitted are denoted as $\textbf{x}_1,\ldots,\textbf{x}_L$ correspondingly, where $\textbf{x}_l\in\mathbb{R}^n,l = 1,\ldots,L$, and $n$ is the codeword length. The transmission powers per dimension are $P_1,\ldots,P_L$ for each user respectively. The receiver wants to recover a linear combination over the finite field, i.e., $\bigoplus_{l=1}^L q_l \textbf{w}_l$, where $q_l$ are the coefficients taking values in $\mathbb{F}_p$. It has been shown in \cite{Nazer11CF} that with nested lattice codes and lattice decoding, this linear combination can be recovered by decoding an integer linear combination of the transmitted signals, i.e., $\sum_{l=1}^L a_l \textbf{x}_l$, where $q_l = g(a_l\mod p)$, and $g$ is the function mapping the elements in $\{0,1,\ldots,p-1\}$ to the corresponding elements in $\mathbb{F}_p$. Based on the original CF scheme in \cite{Nazer11CF}, a generalized scheme in \cite{Zhu17} introduced an extension to the lattice coding scheme, characterized by parameters $\boldsymbol{\beta} = (\beta_1,\ldots,\beta_L)^T$ to allow unequal computation rates for different users. In this case, the achievable computation rate of each user $l\in\{1,\ldots,L\}$ with channel coefficients $\textbf{h} = (h_1,\ldots,h_L)^T$ and decoding coefficients $\textbf{a} = (a_1,\ldots,a_L)^T$ is given by 
\begin{equation}\label{eq_rate_expression_apar_l}
    R_{comp}^{(l)}(\textbf{h},\textbf{a},\textbf{p},\boldsymbol{\beta}) = \frac{1}{2}\log_2^+\left(\beta_l^2\left(||\tilde{\textbf{a}}||^2-\frac{((\sqrt{\textbf{p}}\circ\textbf{h})^T\tilde{\textbf{a}})^2}{1+||\sqrt{\textbf{p}}\circ\textbf{h}||^2}\right)^{-1}\right) \text{ bpcu},
\end{equation}
where $\circ$ is the Hadamard product, $\tilde{\textbf{a}} = \textbf{a}\circ\boldsymbol{\beta}$, $\sqrt{\textbf{p}} = (\sqrt{P_1},\ldots,\sqrt{P_L})^T$, $f^+(\cdot) = \max(0,f(\cdot))$ and bpcu is short for bits per channel use. For brevity, we will assume all rate expressions have the same unit bpcu and use $\log(\cdot)$ to stand for $\log_2(\cdot)$ in the rest of the paper unless otherwise specified. 

\subsection{Computation rate in the fading channel}
In this work, we are interested in understanding the above computation rates in a fading scenario. To start with a simpler problem and keep the fairness of the users in the transmission, we will assume $\beta_l = 1$ for $l=1,\ldots,L$. Then every user has the same achievable computation rate given by 
\begin{equation}\label{eq_rate_expression_apsr}
    R_{comp}(\textbf{h},\textbf{a},\textbf{p}) = \frac{1}{2}\log^+\left(\left(||\textbf{a}||^2-\frac{((\sqrt{\textbf{p}}\circ\textbf{h})^T\textbf{a})^2}{1+||\sqrt{\textbf{p}}\circ\textbf{h}||^2}\right)^{-1}\right).
\end{equation}

Following a standard model for wireless channels (see, e.g. \cite[Chapter 23]{Gamal2011}), we assume a fast fading channels where the coherence time is smaller than the codeword length and we also assume the transmitters and the receiver have the channel state information (CSI). It is assumed that the channel coefficients $h_l$'s are independent for each user $l$. Moreover, we assume the decoding coefficient $\textbf{a}$ is pre-determined and known to all the transmitters and the receiver. This means the decoding purpose is a given integer linear combination. We will then study the achievable computation rate defined in Definition~\ref{def_average_rate} with a power allocation policy denoted by $\textbf{p}(\textbf{h}) = (P_1(\textbf{h}),\ldots,P_L(\textbf{h}))^T$. 

\begin{define}\label{def_average_rate}
For a fading MAC with channel gain $\textbf{\emph{h}}$ where $\textbf{\emph{h}}$ is a random vector, the achievable computation rate with a power allocation policy $\textbf{\emph{p}}(\textbf{\emph{h}})$ and decoding coefficient \textbf{\emph{a}} is given by
\begin{equation}\label{eq_rate_expression_apsr_pa}
    R_{comp}^{fading}(\textbf{\emph{h}},\textbf{\emph{a}},\textbf{\emph{p}}(\textbf{\emph{h}})) = \E_{\textbf{\emph{h}}}\left[\frac{1}{2}\log^+\left(\left(||\textbf{\emph{a}}||^2-\frac{((\sqrt{\textbf{\emph{p}}(\textbf{\emph{h}})}\circ\textbf{\emph{h}})^T\textbf{\emph{a}})^2}{1+||\sqrt{\textbf{\emph{p}}(\textbf{\emph{h}})}\circ\textbf{\emph{h}}||^2}\right)^{-1}\right)\right].
\end{equation}
\end{define}

The basic idea to achieve rate (\ref{eq_rate_expression_apsr_pa}) is to construct different codebooks for different channel states. Similar to the process in \cite[Chapter 7.4]{Gamal2011}, we construct independent codebooks for each sub-channels, where the rate of the codebooks are matched to the computation rate of each sub-channel. The messages are split and encoded by the corresponding codebooks and the codewords are sent depending on the sub-channels. The decoder will decompose the received sequence according to the sub-channels and do the lattice decoding. Based on the argument in \cite{Nazer11CF}, with any message rate no larger than the achievable computation rate of each sub-channel, the error probability to decode each sub-sequence becomes arbitrarily small when the sequence length grows large enough. Thus, the total error probability tends to zero as the sequence length approaches infinity. In addition, with CSIT, the transmitters can also allocate the transmission power according to the sub-channels. The users are allowed to have their own power allocation policy while they are assumed to have the same average transmitted power constraint $\bar P$, i.e., $\E_{\textbf{h}}[P_l(\textbf{h})] \leq \bar P,\forall l \in \{1,\ldots,L\}$. Then the problem of maximizing the achievable computation rate for the fading MAC can be formulated as follows.
\begin{equation}\label{problem_general_apsr}
    \begin{split}
        \max_{\textbf{\textbf{p}}(\textbf{h})} \quad & R_{comp}^{fading},\\
        s.t. \quad & \E_{\textbf{h}}[P_l(\textbf{h})] \leq \bar P,\; l = 1,\ldots,L,\\
        & P_l(\textbf{h}) \geq 0,\; l = 1,\ldots,L.
    \end{split}
\end{equation}

To fully solve this problem, we need to find $L$ power allocation policies, one for each user. The power policies cannot be simply decoupled due to the structure of the optimization problem. Moreover, the expression in (\ref{eq_rate_expression_apsr_pa}) is non-convex and non-smooth in the power policy $\textbf{p}(\textbf{h})$, which makes the optimization problem even more difficult if we consider its optimal solution. In the next section, we will first look at the discrete case where $\textbf{h}$ will take a finite number of values.

\begin{remark}
    We note that the formulated power allocation problem may not fully exploit the CF scheme with CSIT and the limitations are summarized in the following points.
    \begin{itemize}
        \item We only consider equal computation rate in this case. Besides the reason that this assumption makes the problem tractable, it is a more sensible choice from a fairness point of view. If we allow different computation rates and say, maximize the sum computation rate, it often happens that assigning non-zero rate to only one user  (hence effectively no computation involved) is optimal, making the problem non-interesting.
        \item We assume that the coefficient a is given. It is possible to extend the problem where the coefficient $\textbf{\emph{a}}$ can be chosen (optimized) by the receiver with additional optimization procedures. This is not explored in the current paper.
        \item The strategy is based on separate coding for each sub-channel and does not consider jointly coding on all channel uses, for example, constructing codebooks across all sub-channels (see \cite{zhang24}). 
    \end{itemize}
    
    However, we will see later in this paper, the problem formulated with equal rate expression and only power optimisation is already a non-trivial problem. Especially, we focus on analysing the optimal solution of the CF power allocation problem, which is not thoroughly studied in the previous literature. The derived results for this simplified case provide important insights on the properties of the optimal solutions and connect the classical water-filling algorithm to this non-convex problem. We believe these results are instructive for a more general problem formulation, and importantly, the proposed order-based algorithms can be used as a framework or building block when studying more general problems.
\end{remark}

\section{Discrete case}

We first assume $h_1,\ldots,h_L$ are discretely distributed and independent real random variables. The channel $h_i$, $i = 1,\ldots,L$, takes value $h_{ij}$ with probability $p_{ij}$, $j = 1,\ldots,l_i$. There are in total $M = \prod_{i=1}^L l_i$ channel states denoted as $\textbf{h}_1,\ldots,\textbf{h}_M \in \mathbb{R}^L$, and their probability $f_m$ can be calculated by multiplying the corresponding $p_{ij}$'s. The target optimization problem is given by
\begin{equation}\label{primal_problem_discrete_apsr}
    \begin{split}
        DP1:\quad\max_{\textbf{P}} \quad & \sum_{m=1}^M f_m R_{comp}(\textbf{h}_m,\textbf{a},\textbf{p}_m),\\
        s.t. \quad & \sum_{m=1}^M f_m P_{lm} \leq \bar P, \; l = 1,\ldots,L,\\
        & P_{lm} \geq 0,\quad l = 1,\ldots,L,\; m = 1,\ldots,M,
    \end{split}
\end{equation}
where $\textbf{P} = (\textbf{p}_1,\ldots,\textbf{p}_M)$, $\textbf{p}_m = (P_{1m},\ldots,P_{Lm})^T, m = 1,\ldots,M$, $P_{lm}$ is the transmission power for user $l$ at channel state $m$ and $R_{comp}$ is defined in (\ref{eq_rate_expression_apsr}). We further define $R_m(\textbf{p}_m) = \frac{1}{2}\log\left(\left(||\textbf{a}||^2-\frac{((\sqrt{\textbf{p}_m}\circ\textbf{h}_m)^T\textbf{a})^2}{1+||\sqrt{\textbf{p}_m}\circ\textbf{h}_m||^2}\right)^{-1}\right)$, so that $R_{comp}(\textbf{h}_m,\textbf{a},\textbf{p}_m) = R_m^+(\textbf{p}_m),\forall m$.

\subsection{Theoretical analysis of the optimization problem}
\subsubsection{\textbf{General properties of the optimal solution}}
The problem $DP1$ is a non-convex problem with a non-smooth objective function, which does not admit a simple solution either analytically or numerically. The following definition and auxiliary results help us to understand the properties of the optimal solutions. Note that all norms in the paper refer to the $\ell_2$ norm if it is a vector (instead of a function). We call a power allocation policy a feasible policy if it satisfies the power constraints in (\ref{primal_problem_discrete_apsr}).

\begin{define} (Active set) The active set $\mathcal{A}_{\textbf{\emph{P}}}$ of a feasible solution $\textbf{\emph{P}}$ of the problem $DP1$ is the set of channel states in which at least one user transmits with positive power, i.e., $\mathcal{A}_{\textbf{\emph{P}}} = \left\{m\in\{1,\ldots,M\}\Big|||\textbf{\emph{p}}_m|| > 0 \right\}$.
\end{define}

\begin{define} (Optimal active set)
Given an optimal solution $\textbf{\emph{P}}^*$, we use $\mathcal{A}_{\textbf{\emph{P}}^*}$ to denote the optimal active set associated to this solution. For simplicity we also use $\mathcal{A}^*$ to denote an optimal active set with the understanding that there is an associated optimal solution. 
\end{define}

\begin{remark}(The uniqueness of the optimal solution)\\
In some cases, the global optimal solution as well as the optimal active set is not unique. One example is provided as follows. Assume both users have two possible channel coefficients $0.5$ and $1$ with equal probability. There are totally four possible channel states, namely $\textbf{\emph{h}}_1=(0.5,0.5)^T, \textbf{\emph{h}}_2=(0.5,1)^T,\textbf{\emph{h}}_3=(1,0.5)^T,\textbf{\emph{h}}_4=(1,1)^T$ with equal probability $0.25$. We set the computational coefficients $\textbf{\emph{a}} = (1,1)^T$ and power constraint $\bar P = 2$. The optimal active set(s) will among all the $2^4 = 16$ possible active sets. It can be verified that when the active set $\mathcal{A} = \{2,4\}$ and $\{3,4\}$, the power policies give the same optimal average rate ($0.4102$). Their power policies are given by $\textbf{\emph{P}} = \left(\begin{array}{cccc}
    0 & 3.3896 & 0 & 4.6105 \\
    0 & 2.5853 & 0 & 5.4145
\end{array}\right)$ and $\left(\begin{array}{cccc}
    0 & 0 & 2.5853 & 5.4145 \\
    0 & 0 & 3.3896 & 4.6105
\end{array}\right)$, respectively. Thus, $\{2,4\}$ and $\{3,4\}$ are both the optimal active sets in this case, and the optimal solution is not unique.
\end{remark}

Given an active set $\mathcal{A} = \mathcal{A}_{\textbf{P}}$ for some feasible power allocation policy $\textbf{P}$, the channel states outside $\mathcal{A}$ (with zero power) have no contribution to the average rate. However, the channel states inside $\mathcal{A}$ (with positive power) do not guarantee a positive rate because the function $R_m(\textbf{p})$ can be negative given some nonzero $\textbf{p}$. In other words, some channel states may consume certain positive power without contributing to the average rate if the power policy is not constructed carefully. Note that this may happen for some feasible power allocation policies while not possible for an optimal policy, as shown in the following lemma.

\begin{Lem}\label{lem_optimal_active_set}
If $\mathcal{A}^*$ is an optimal active set of $DP1$ and $\textbf{\emph{P}}^*$ is the corresponding solution, then all the channel states in $\mathcal{A}^*$ have positive rate, i.e., $R_m(\textbf{\emph{p}}_m^*) > 0, \forall m \in \mathcal{A}^*$.
\end{Lem}
\begin{proof}
Assume there is an optimal active set $\mathcal{A}^*$ with the corresponding optimal power allocation policy $\textbf{P}^*$ such that it includes a channel state $i$ with $R_i(\textbf{p}_i^*) \leq 0$. Since $i\in \mathcal{A}^*$, we have $||\textbf{p}_i^*|| > 0$. Denote the average rate in this case as $R_{avg}$. We now show that there exists another power allocation policy that gives a higher average rate. First we set the power allocation policy of channel $i$ to be zero, i.e., $\textbf{p}_i^\prime = \textbf{0}$. Then we re-allocate $\textbf{p}_i^*$ to another channel state $j\in \mathcal{A}^*$ with $R_j(\textbf{p}_j^*) > 0$. We construct its new power allocation policy $\textbf{p}_j^\prime$ such that $\textbf{p}_j^\prime = \alpha \textbf{p}_j^*$ with some $\alpha>1$. Notice that such an $\alpha$ exists as long as $(\alpha-1)\textbf{p}_j^* \leq \textbf{p}_i^*$ in an element-wise sense, which keep the power constraint satisfied. Thus, 
\begin{align*}
    R_j(\textbf{p}_j^\prime) & = \frac{1}{2}\log\left(\left(||\textbf{a}||^2-\frac{((\sqrt{\textbf{p}_j^\prime}\circ\textbf{h}_j)^T\textbf{a})^2}{1+||\sqrt{\textbf{p}_j^\prime}\circ\textbf{h}_j||^2}\right)^{-1}\right)\\
    & = \frac{1}{2}\log\left(\left(||\textbf{a}||^2-\frac{\alpha((\sqrt{\textbf{p}_j^*}\circ\textbf{h}_j)^T\textbf{a})^2}{1+\alpha ||\sqrt{\textbf{p}_j^*}\circ\textbf{h}_j||^2}\right)^{-1}\right) \\
    & \stackrel{(a)}{>} \frac{1}{2}\log\left(\left(||\textbf{a}||^2-\frac{((\sqrt{\textbf{p}_j^*}\circ\textbf{h}_j)^T\textbf{a})^2}{1+ ||\sqrt{\textbf{p}_j^*}\circ\textbf{h}_j||^2}\right)^{-1}\right) = R_j(\textbf{p}_j^*)
\end{align*}
The inequality $(a)$ holds since $\frac{\alpha x}{1+\alpha y}>\frac{x}{1+ y}$ for $x,y>0$ and $\alpha>1$. Denote the average rate in this case as $R_{avg}^\prime$. Comparing with the previous rate $R_{avg}$, we have 
\begin{equation*}
    R_{avg}^\prime - R_{avg} = \sum_{m=i,j} f_m (R_m^+(\textbf{p}_m^\prime)-R_m^+(\textbf{p}_m^*)) > 0
\end{equation*}
This shows $\textbf{P}^*$ is not optimal hence Lemma~\ref{lem_optimal_active_set} is proved.
\end{proof} 

Given Lemma~\ref{lem_optimal_active_set}, the next Lemma shows that $DP1$ can be simplified if the optimal active set is known.

\begin{Lem}\label{lem_simple_prob}
Given an optimal active set $\mathcal{A}^*$, the corresponding optimal power allocation policy $\textbf{\emph{P}}^*$ satisfies $\textbf{\emph{p}}_m^* = \textbf{0}$ for all $m\notin\mathcal{A}^*$. Furthermore, the optimal solution $\textbf{\emph{p}}_m^*$ with $m\in\mathcal{A}^*$ is the solution to the following optimization problem with $\mathcal{S} = \mathcal{A}^*$.
\begin{equation}\label{convex_problem_discrete_apsr}
    \begin{split}
        DP2: \quad \max_{\{\textbf{\emph{p}}_m\}_{m\in \mathcal{S}}} \quad &  \sum_{m\in \mathcal{S}}
        f_m R_m(\textbf{\emph{p}}_m)
        ,\\
        s.t. \quad & \sum_{m\in\mathcal{S}} f_m P_{lm} \leq \bar P, \; l = 1,\ldots,L,\\
        & P_{lm} \geq 0,\quad l = 1,\ldots,L,\; m \in \mathcal{S},
    \end{split}
\end{equation}
where $\mathcal{S}$ is the support set of $DP2$, and $\mathcal{S}\subseteq\{1,\ldots,M\}$.
\end{Lem}
\begin{proof}
If we know an optimal active set, from Lemma~\ref{lem_optimal_active_set}, $R_m(\textbf{p}_m^*) > 0$, hence $R_m^+(\textbf{p}_m^*) = R_m(\textbf{p}_m^*)$ for $m \in \mathcal{A}^*$. Thus, we can release the $\max(0,\cdot)$ function for those channel $m$. Note that for other channel state $n\notin\mathcal{A}^*$, we have $\textbf{p}_n^* = \textbf{0}$ and $R_n^+(\textbf{p}_n^*) = 0$. Therefore we can further exclude those channels in the optimization problem, which leads to $DP2$. 
\end{proof}

Compared to the target problem $DP1$, the problem $DP2$ has a smooth objective function which is in general easier to deal with. If we can identify the optimal active set $\mathcal{A}^*$, we can solve $DP1$ by solving $DP2$. However, it turns out that finding the optimal set $\mathcal{A}^*$ is a highly non-trivial problem. In order to motivate the idea of how to determine the optimal active set, we will first look at a special case in the next section.

\subsubsection{\textbf{Guarantee for the symmetric power allocation policy}}
In this subsection, we consider a special case where all users share the same power allocation policy, i.e., $\textbf{p}_m = P_m\times\textbf{1}_L$ for $m=1,\ldots,M$, where $\textbf{1}_L$ denotes a $L$-dimensional vector with all entries equaling one and $P_m$ is the common allocated power for all users at channel state $m$. This simplification might also be of practical interest because the power allocation policy is easier to compute. In this case, the function $R_m(\textbf{p}_m)$ is given by $R_m(\textbf{p}_m) = R_m(P_m) = \frac{1}{2}\log\frac{1+P_m ||\textbf{h}_m||^2}{ ||\textbf{a}||^2 + P_m \left( ||\textbf{h}_m||^2||\textbf{a}||^2 - (\textbf{h}_m^T\textbf{a})^2\right)}, m = 1,\ldots,M$. The target problem $DP1$ will be simplified to 
\begin{equation}\label{primal_problem_discrete}
    \begin{split}
        DP1^{(s)}:\quad\max_{P_1,\ldots,P_M} \quad & \sum_{m=1}^M f_m R_m^+(P_m),\\
        s.t. \quad & \sum_{m=1}^M f_m P_m \leq \bar P,\\
        & P_m \geq 0,\quad m = 1,\ldots,M.
    \end{split}
\end{equation}
If the optimal active set is known, we have the following problem corresponding to the problem $DP2$ in Lemma~\ref{lem_simple_prob} by setting $\mathcal{S}$ to be the optimal active set.
\begin{equation}\label{convex_problem_discrete}
    \begin{split}
        DP2^{(s)}: \quad \max_{\{P_m\}_{m\in \mathcal{S}}} \quad &  \sum_{m\in \mathcal{S}}
        f_m R_m(P_m)
        ,\\
        s.t. \quad & \sum_{m\in \mathcal{S}} f_m P_{m} \leq \bar P, \\
        & P_{m} \geq 0,\quad l = 1,\ldots,L.
    \end{split}
\end{equation}

Note that the above problem is convex and satisfies Slater's condition. Thus, the optimal solution can be given by Karush–Kuhn–Tucker (KKT) conditions\cite{Boyd04ConvexOptimization}. We present its optimal solution in an analytical form in Lemma~\ref{lem_optimal_solution_simple}. Note that we say two nonzero vectors $\textbf{x}_1,\textbf{x}_2$ are collinear if they satisfy $\textbf{x}_1 = c \textbf{x}_2$ for some real number $c$.

\begin{Lem}\label{lem_optimal_solution_simple}
The optimal solutions $P_m^*$ of the problem $DP2^{(s)}$ is given by
\begin{equation}\label{eq_KKT_solution}
    P_m^* = (P_m^{KKT}(\lambda^*))^+,m\in \mathcal{S}
\end{equation}
with some $\lambda^*>0$ satisfying the equation
\begin{equation}\label{eq_sum_constraint}
    \sum_{m\in\mathcal{S}} f_m P_m^* = \bar P.
\end{equation}
$P_m^{KKT}$ is defined as
\begin{equation}\label{eq_P_KKT}
    P_m^{KKT}(\lambda) = \Bigg\{\begin{array}{ll}
        \frac{1}{\lambda}-\frac{1}{||\textbf{\emph{h}}_m||^2}, & \textbf{\emph{h}}_m \text{ is collinear with } \textbf{\emph{a}} \\
        \frac{ - b_m + \sqrt{b_m^2 - 4 d_m c_m(\lambda)}}{2d_m}, & \textbf{\emph{h}}_m \text{ is not collinear with } \textbf{\emph{a}}
    \end{array}
\end{equation}
with $d_m = ||\textbf{\emph{h}}_m||^2 (||\textbf{\emph{h}}_m||^2||\textbf{\emph{a}}||^2-(\textbf{\emph{h}}_m^T\textbf{\emph{a}})^2)$, $b_m = 2||\textbf{\emph{h}}_m||^2||\textbf{\emph{a}}||^2-(\textbf{\emph{h}}_m^T\textbf{\emph{a}})^2$, $c_m(\lambda) = ||\textbf{\emph{a}}||^2 - \frac{(\textbf{\emph{h}}_m^T\textbf{\emph{a}})^2}{\lambda}$.
\end{Lem}

\begin{proof}
We will first show the problem $DP2^{(s)}$ is a convex problem. Firstly, the constraint space is an intersection of halfspaces, thus convex. Secondly, the second derivative of $R_m(P_m)$ is given by
\begin{equation}
    R_m^{''}(P_m) = - \frac{(\textbf{h}_m^T\textbf{a})^2(||\textbf{h}_m||^2||\textbf{a}||^2+(1+2P_m||\textbf{h}_m||^2)(||\textbf{h}_m||^2||\textbf{a}||^2-(\textbf{h}_m^T\textbf{a})^2))}{(||\textbf{a}||^2+P_m||\textbf{h}_m||^2)^2(1+P_m||\textbf{h}_m||^2)^2} < 0, 
\end{equation}
for any $P_m \geq 0$, which implies $R_m(P_m)$ is strictly concave. Here we use Cauchy–Schwarz inequality to show that $||\textbf{h}_m||^2||\textbf{a}||^2-(\textbf{h}_m^T\textbf{a})^2 \geq 0$. Then the objective function of $DP2^{(s)}$ is strictly concave since it is a linear combination of $R_m(P_m)$ with all positive coefficients $f_m$, for $m \in \mathcal{S}$. Therefore, $DP2^{(s)}$ is a convex optimization problem.

The Lagrangian function of the objective of the problem $DP2^{(s)}$ is given by
\begin{equation}
    L(\textbf{P},\lambda) = \left(\sum_{m\in \mathcal{S}} -f_m R_m(P_m)\right) + \lambda \left(\sum_{m\in \mathcal{S}} f_m P_m - \bar P \right) 
\end{equation}
with complementary constraint $P_m\geq 0$, where $\lambda\geq 0$ is the Lagrange multiplier. By setting $\frac{\partial L(\textbf{P},\lambda)}{\partial P_m} = 0$, we obtain
\begin{equation}
    -f_mR_m^{'}(P_m) + \lambda f_m = 0.
\end{equation}
Calculating $R_m^{'}(P_m)$ and plugging it into the above equation, we have
\begin{equation}\label{eq_op_sol_T3_fraction}
    \frac{(\textbf{h}_m^T\textbf{a})^2}{||\textbf{h}_m||^2 (||\textbf{h}_m||^2||\textbf{a}||^2-(\textbf{h}_m^T\textbf{a})^2) P_m^2 + (2||\textbf{h}_m||^2||\textbf{a}||^2-(\textbf{h}_m^T\textbf{a})^2) P_m + ||\textbf{a}||^2} = 2 \lambda.
\end{equation}
Since the left hand side (LHS) of (\ref{eq_op_sol_T3_fraction}) is nonzero for nonzero $\textbf{h}_m$, it implies that $\lambda > 0$. Here we replace $2\lambda$ with $\lambda$ without loss of generality. The equation (\ref{eq_op_sol_T3_fraction}) can be further simplified as
\begin{equation}\label{eq_optimal_solution_T3}
    \underbrace{||\textbf{h}_m||^2 (||\textbf{h}_m||^2||\textbf{a}||^2-(\textbf{h}_m^T\textbf{a})^2)}_{:=d_m} P_m^2 + \underbrace{(2||\textbf{h}_m||^2||\textbf{a}||^2-(\textbf{h}_m^T\textbf{a})^2)}_{:=b_m} P_m + \underbrace{||\textbf{a}||^2 - \frac{(\textbf{h}_m^T\textbf{a})^2}{\lambda} }_{:=c_m(\lambda)} = 0.
\end{equation}
Recall that the channel coefficients are real, thus $d_m, b_m$ and $c_m(\lambda)$ are all real. When $\textbf{h}$ and $\textbf{a}$ are collinear, we have $d_m = 0$. In this case the solution of (\ref{eq_optimal_solution_T3}) is given by
\begin{equation}
    P_m = -\frac{c_m(\lambda)}{b_m} = \frac{1}{\lambda}-\frac{1}{||\textbf{h}_m||^2}.
\end{equation} 
When $\textbf{h}$ and $\textbf{a}$ are not collinear, the solution of (\ref{eq_optimal_solution_T3}) is given by 
\begin{equation}
    P_m = \frac{ - b_m \pm \sqrt{b_m^2 - 4 d_m c_m(\lambda)}}{2d_m}.
\end{equation}
Since $d_m,b_m>0$, the solution with minus sign will be negative in any case, thus should be discarded. Recalling the complementary constraint $P_m \geq 0$, and based on the fact that $R_m(P_m)$ is increasing with $P_m$, we have the solution form as shown in Lemma~\ref{lem_optimal_solution_simple}.
\end{proof}

It can be observed that when $\textbf{h}$ and $\textbf{a}$ are collinear, the solution actually reduces to the classical water-filling solution\cite{Goldsmith97}. Since $c_m(\lambda)$ is monotonically increasing with $\lambda$, we note that $P_m^{KKT}(\lambda)$ is always a monotonically decreasing function of $\lambda$. The result in Lemma~\ref{lem_optimal_solution_simple} is reminiscent of the well-known water-filling algorithm, where $\lambda$ determines the water level. The water level will decide which channel states are active. Therefore, in practice we can apply bisection algorithm to determine the optimal $\lambda$ as shown in Algorithm~\ref{algo_optimal_simple}.

\begin{algorithm}
\caption{Optimal algorithm for $DP2^{(s)}$}\label{algo_optimal_simple}
\begin{algorithmic}[1]
\scriptsize
\State Define $q(\lambda) = \sum_{m\in\mathcal{S}} f_m (P_m^{KKT}(\lambda))^+$ \Comment{Note that the power constraint is $q(\lambda) \leq \bar P$}
\State Choose a small enough $\lambda_l$ (e.g. $\lambda_l<10^{-3}$) such that $q(\lambda_l) > \bar P$
\State Choose a large enough $\lambda_u$ (e.g. $\lambda_u>10^{3}$) such that $q(\lambda_u) < \bar P$
\State $\lambda \gets \frac{\lambda_l+\lambda_u}{2}$
\While{$\bar P-q(\lambda)>\epsilon \;\big|\big|\; \bar P - q(\lambda)<0$} \Comment{$\epsilon$ is a small number, e.g. $10^{-3}$}
\If{$\bar P-q(\lambda)>0$}
    \State $\lambda_u \gets \lambda$
\Else
    \State $\lambda_l \gets \lambda$
\EndIf
\State $\lambda \gets \frac{\lambda_l+\lambda_u}{2}$
\EndWhile
\State $P_m^* = (P_m^{KKT}(\lambda))^+,\forall m\in\mathcal{S}$
\end{algorithmic}
\end{algorithm}

As mentioned in the previous subsection, finding the optimal active set $\mathcal{A}^*$ is a highly non-trivial problem. In this simplified case, we study the property of $\mathcal{A}^*$ by introducing the following two definitions.

\begin{define}\label{def_good_set} (Good set) The good set $\mathcal{G}$ is the subset of $\{1,\ldots,M\}$ which include all $m$ with $||\textbf{\emph{h}}_m||^2 > ||\textbf{\emph{h}}_m||^2||\textbf{\emph{a}}||^2 - (\textbf{\emph{h}}_m^T\textbf{\emph{a}})^2$, i.e., $\mathcal{G} = \left\{m\in\{1,\ldots,M\}\Big| ||\textbf{\emph{h}}_m||^2 > ||\textbf{\emph{h}}_m||^2||\textbf{\emph{a}}||^2-(\textbf{\emph{h}}_m^T\textbf{\emph{a}})^2\right\}$.
\end{define}

\begin{define}\label{def_bad_set} (Bad set) The bad set $\mathcal{B}$ is the complement of the good set in $\{1,\ldots,M\}$, i.e., $\mathcal{B} = \{1,\ldots,M\} \backslash \mathcal{G}$.
\end{define}

The following theorem will present the relationship between $\mathcal{A}^*$ and these two sets.

\begin{theorem}\label{thm_noc1}
Let $\mathcal{A}^*$ denote an optimal active set for the symmetric power allocation policy. It always holds that $\mathcal{A}^*\cap \mathcal{B}=\emptyset$ and $\mathcal{A}^*\subseteq \mathcal{G}$. 
\end{theorem} 
\begin{proof}
Let $P_1^*, \ldots, P_M^*$ denote the optimal power allocations corresponding to $\mathcal{A^*}$. Assume $\mathcal{A^*}$ includes some channel state $k\in\mathcal{B}$. By the definition of the bad set, $||\textbf{h}_k||^2\leq ||\textbf{h}_k||^2||\textbf{a}||^2 - (\textbf{h}_k^T\textbf{a})^2$. Since $\textbf{a}$ is a nonzero integer vector, we have $||\textbf{a}||\geq1$. Then we have $1+P_k ||\textbf{h}_k||^2 \leq ||\textbf{a}||^2 +  P_k(||\textbf{h}_k||^2||\textbf{a}||^2 - (\textbf{h}_k^T \textbf{a})^2),\forall P_k\geq 0$. Recalling the expression of $R_m(P_m)$, we have $R_k(P_k)\leq0, \forall P_k\geq 0$, especially $R_k(P_k^*)\leq0$. However, from Lemma~\ref{lem_optimal_active_set}, any channel state in the optimal active set will return a positive rate. Thus, a contradiction happens. Therefore $\mathcal{A}^*$ should not include any channel state in the bad set. By the relationship between the good set and the bad set, the optimal active set is a subset of the good set.
\end{proof} 

Furthermore, the following result states that when the power constraint $\bar P$ is large enough, the optimal active set $\mathcal{A}^*$ is equal to the good set $\mathcal{G}$.

\begin{theorem}\label{thm_opt_large_p}
The optimal active set of the problem $DP1^{(s)}$ is the good set, i.e., $\mathcal{A}^* = \mathcal{G}$ when the power constraint $\bar P$ is larger than $\bar P_o = \sum_{m\in\mathcal{G}} f_m  P_m^{KKT}(\lambda_o)$ with
$$\lambda_o = \min \left\{\lambda \Big| R_m(P_m^{KKT}(\lambda)) = R_m^\prime(P_m^{KKT}(\lambda))P_m^{KKT}(\lambda), m\in\mathcal{G}\right\},$$ 
where $R_m(\cdot)$ is given in (\ref{primal_problem_discrete}) and $P_m^{KKT}(\lambda)$ is given by (\ref{eq_P_KKT}).
\end{theorem}

\begin{proof}
Assume $\mathcal{A}^*\neq\mathcal{G}$ when $\bar P>\bar P_o$. From Theorem \ref{thm_noc1} $\mathcal{A}^*\subseteq\mathcal{G}$, then $\mathcal{A}^*\subset\mathcal{G}$. Denote $\mathcal{A}^* = \mathcal{G}\backslash\mathcal{K}$, where $\mathcal{K}\subset \mathcal{G}$ is some non-empty set. From Lemma~\ref{lem_simple_prob} and Lemma~\ref{lem_optimal_solution_simple}, since $\mathcal{A}^*$ is given, the optimal solution of the problem $DP1^{(s)}$ can be given by Algorithm~\ref{algo_optimal_simple} with $\mathcal{S} = \mathcal{A}^*$. Denote the optimal Lagrangian multiplier as $\lambda_k$. Then the nonzero power allocation is given by $P_m^{KKT}(\lambda_k),m\in\mathcal{A}^*$, and the optimal average rate is $\sum_{m\in\mathcal{A}^*} f_m R_m(P_m^{KKT}(\lambda_k))$. Here, we omit the $\max(0,\cdot)$ function in the rate expression from Lemma~\ref{lem_optimal_active_set}. In addition, the power constraint gives $\sum_{m\in\mathcal{A}^*} f_m P_m^{KKT}(\lambda_k) = \bar P$.

We can also derive a non-optimal solution by applying Algorithm~\ref{algo_optimal_simple} with $\mathcal{S} = \mathcal{G}$. Denote the Lagrangian multiplier of this solution as $\lambda_g$. Then the power allocation is given by $P_m^{KKT}(\lambda_g)$ for $m\in\mathcal{G}$ and zero for other channel states. The corresponding average rate is $\sum_{m\in\mathcal{G}} f_m R_m^+(P_m^{KKT}(\lambda_g))$, and the power constraint gives $\sum_{m\in\mathcal{G}} f_m P_m^{KKT}(\lambda_g) = \bar P$.

From the optimality of $\mathcal{A}^*$, we have $\sum_{m\in\mathcal{A}^*} f_m R_m(P_m^{KKT}(\lambda_k)) \geq \sum_{m\in\mathcal{G}} f_m R_m^+(P_m^{KKT}(\lambda_g))$. Since $\max(0,x)\geq x,\forall x$, we have $\sum_{m\in\mathcal{A}^*} f_m R_m(P_m^{KKT}(\lambda_k)) \geq \sum_{m\in\mathcal{G}} f_m R_m(P_m^{KKT}(\lambda_g))$. By moving terms, the following is observed.
\begin{equation}\label{eq_ineq_a_k}
    \sum_{m\in\mathcal{A}^*} f_m \big(R_m(P_m^{KKT}(\lambda_k)) - R_m(P_m^{KKT}(\lambda_g))\big) \geq \sum_{m\in\mathcal{K}} f_m R_m(P_m^{KKT}(\lambda_g)).
\end{equation}
From the power constraint, $\sum_{m\in\mathcal{A}^*} f_m P_m^{KKT}(\lambda_k) = \sum_{m\in\mathcal{G}} f_m P_m^{KKT}(\lambda_g)$. As $\mathcal{A}^*\subset\mathcal{G}$, $P_m^{KKT}(\lambda_k) \geq P_m^{KKT}(\lambda_g),\forall m\in\mathcal{A}^*$. This is because $P_m^{KKT}(\lambda)$ is monotonically increasing with $\lambda$. By moving terms of the power equation, we have
\begin{equation}\label{eq_power_a_k}
    \sum_{m\in\mathcal{A}^*} f_m (P_m^{KKT}(\lambda_k) - P_m^{KKT}(\lambda_g)) = \sum_{m\in\mathcal{K}} f_m P_m^{KKT}(\lambda_g)
\end{equation}
As $R_m(P_m)$ is concave, it holds that $R_m(x_2) - R_m(x_1) \leq R_m^\prime(x_1)(x_2-x_1), \forall x_2\geq x_1 \geq 0$. Then $R_m(P_m^{KKT}(\lambda_k)) - R_m(P_m^{KKT}(\lambda_g)) \leq R_m^\prime(P_m^{KKT}(\lambda_g))(P_m^{KKT}(\lambda_k)-P_m^{KKT}(\lambda_g)),\forall m\in\mathcal{A}^*$. Thus, combining (\ref{eq_ineq_a_k}), we have
\begin{equation}\label{eq_ineq_concave}
    \sum_{m\in\mathcal{A}^*} f_m R_m^\prime(P_m^{KKT}(\lambda_g))(P_m^{KKT}(\lambda_k)-P_m^{KKT}(\lambda_g)) \geq \sum_{m\in\mathcal{K}} f_m R_m(P_m^{KKT}(\lambda_g)).
\end{equation}
Note that $R_m^\prime(P_m^{KKT}(\lambda_g)) = \lambda_g,\forall m\in\mathcal{G}$, then it can be lifted out of the summation on the left side of (\ref{eq_ineq_concave}). Then using (\ref{eq_power_a_k}), we can rewrite (\ref{eq_ineq_concave}) as 
\begin{equation}\label{eq_ineq_lift_derivative}
    \lambda_g \sum_{m\in\mathcal{K}} f_m P_m^{KKT}(\lambda_g) \geq \sum_{m\in\mathcal{K}} f_m R_m(P_m^{KKT}(\lambda_g)).
\end{equation}
Placing $\lambda_g$ into the summation of the left side results in
\begin{equation}\label{eq_ineq_reduce_k}
    \sum_{m\in\mathcal{K}} f_m \big(R_m^\prime(P_m^{KKT}(\lambda_g)) P_m^{KKT}(\lambda_g) - R_m(P_m^{KKT}(\lambda_g))\big) \geq 0.
\end{equation}
This indicates that there exists at least one channel state $k\in\mathcal{K}$ such that
\begin{equation}\label{eq_ineq_contradiction_k}
    R_k^\prime(P_k^{KKT}(\lambda_g)) P_k^{KKT}(\lambda_g) - R_k(P_k^{KKT}(\lambda_g)) \geq 0.
\end{equation}

Next we will find out a contradiction to the above statement. Recall that $R_m(x)$ is an increasing concave function for $x\geq 0$ and $m\in\mathcal{G}$, i.e., $R_m^\prime(x)>0$, $R_m^{\prime\prime}(x)<0$. Thus, the function $l_m(x) = R_m(x)-R_m^\prime(x)x$ is an increasing function of $x$ since $l_m^\prime(x) = -R_m^{\prime\prime}(x)>0$. In addition, $l_m(0) = R_m(0) = -\frac{1}{2}$ and $\lim_{x\rightarrow\infty} l_m(x) = \lim_{x\rightarrow\infty} R_m(x)>0$. So there exists only one $x_m>0$ such that $l_m(x_m) = 0$ for each $m\in\mathcal{G}$. Since $P_m^{KKT}(\lambda)$ is a monotonically decreasing function of $\lambda$, there exists only one $\lambda_m>0$ such that $l_m(P_m^{KKT}(\lambda_m)) = 0$ for each $m\in\mathcal{G}$, and $l_m(P_m^{KKT}(\lambda))>l_m(P_m^{KKT}(\lambda_m)) = 0,\forall 0 <\lambda < \lambda_m$. When $\bar P>\bar P_o$, from the expression of $\bar P_o$, we have $\lambda_g<\lambda_o\leq \lambda_m$. Thus, $l_m(P_m^{KKT}(\lambda_g)) > 0$ i.e.,
\begin{equation}\label{eq_ineq_contradiction}
    R_m(P_m^{KKT}(\lambda_g)) - R_m^\prime(P_m^{KKT}(\lambda_g))P_m^{KKT}(\lambda_g) > 0, \forall m\in\mathcal{G}.
\end{equation}
This leads to a contradiction of (\ref{eq_ineq_contradiction_k}). Thus, the assumption $\mathcal{A}^*\neq\mathcal{G}$ breaks. Therefore we can conclude that $\mathcal{A}^* = \mathcal{G}$ when $\bar P>\bar P_o$.
\end{proof}

Theorem 2 shows that when $\bar P>\bar P_o$, the target problem in the symmetric policy case can be solved efficiently with Lemma~\ref{lem_optimal_solution_simple} by setting $\mathcal{A}^*=\mathcal{G}$. In the following, we consider the low to medium power regime where the optimal active set could be a strict subset of $\mathcal{G}$. Moreover, we will discuss how to extend those algorithms to the asymmetric policy case previously formulated in $DP1$. 

\subsubsection{\textbf{Further discussion on the asymmetric power allocation policy scenario}}
We restate the objective function $R_m(\textbf{p}_m)$ when we have a general (asymmetric) power allocation policy.
\begin{equation}\label{eq_R_asm}
    R_m(\textbf{p}_m) = \frac{1}{2}\log\frac{1+||\sqrt{\textbf{p}_m}\circ\textbf{h}_m||^2}{||\textbf{a}||^2+(||\sqrt{\textbf{p}_m}\circ\textbf{h}_m||^2||\textbf{a}||^2-((\sqrt{\textbf{p}_m}\circ\textbf{h}_m)^T\textbf{a})^2)}.
\end{equation}
First notice that unlike the objective function for the symmetric policy, this objective function is non-convex in general. Hence the solution to the KKT conditions cannot guarantee a global optimal solution. Intuitively, the optimal power allocation tries to align the channel coefficients $\textbf{h}_m$ to the equation coefficients $\textbf{a}_m$ as close as possible to minimise the term $||\sqrt{\textbf{p}_m}\circ\textbf{h}_m||^2||\textbf{a}||^2-((\sqrt{\textbf{p}_m}\circ\textbf{h}_m)^T\textbf{a})^2$, hence achieve a larger rate. Particularly, the channel states in the Bad Set can have positive rate with asymmetric power allocation. This is different from the symmetric power policy scenario, and Theorem~\ref{thm_noc1} will not apply to the asymmetric policy scenario in general. If the power constraint is large enough, we naturally expect that all the channel states will be active. However, since we do not have the analytical expressions for the optimal solution in this asymmetric policy scenario, it is difficult to give a threshold like we did in Theorem~\ref{thm_opt_large_p} for the symmetric policy scenario. In the sequel, we will investigate how to determine the active set for a general asymmetric policy, and use off-the-shelf numerical methods to solve the associated optimization problem in $DP2$.

\subsection{Power allocation algorithms for the discrete case}\label{sec_algo_d}

\renewcommand{\thealgorithm}{A\arabic{algorithm}}
\setcounter{algorithm}{-1}

We will first introduce some benchmark algorithms and then propose our iterative algorithm to identify the optimal active set for a general power constraint. We present the algorithms for both symmetric power allocation policy case (special case) and the asymmetric policy case (general case) on the left and the right side of the table, respectively. 

The first benchmark algorithm Algorithm~\ref{algo_const_power} is the trivial power allocation scheme that applies equal power to all users at all channel states.

\begin{algorithm}[h]
\caption{Constant power}\label{algo_const_power}
  \begin{multicols}{2}
    \begin{algorithmic}[1]
      \scriptsize
      \Statex $DP1^{(s)}$:
      \State $P_m = \bar P, \forall m\in\mathcal{G}$.
    \end{algorithmic}
    \columnbreak
    \begin{algorithmic}
      \scriptsize
      \State $DP1$:
      \State $P_{lm} = \bar P, \forall l\in\{1,\ldots,L\}, m\in\{1,\ldots,M\}$.
    \end{algorithmic}
  \end{multicols}
\end{algorithm}

We consider Algorithm~\ref{algo_water_filling}, which directly sets all the possible channel state to be active. Lemma~\ref{lem_optimal_solution_simple} provided a more explicit characterization of the solution for the symmetric policy. The exact expression is different from the classical water-filling solution. But in order to keep the notations simple, we still call this algorithm ``water-filling algorithm''.

\begin{algorithm}[h]
\caption{Water-filling algorithm}\label{algo_water_filling}
  \begin{multicols}{2}
    \begin{algorithmic}[1]
      \scriptsize
      \Statex $DP1^{(s)}$:
      \State Set $\mathcal{S}^{(1)} = \mathcal{G}$, $\{P_m^{(1)}\}_{m=1}^M = \{P_m^{(2)}\}_{m=1}^M = \textbf{0}$;
      \State Run Algorithm~\ref{algo_optimal_simple} with $\mathcal{S} = \mathcal{S}^{(1)}$;
      \State Update solution $\{P_m^{(1)}\}_{m\in \mathcal{S}^{(1)}}$;
      \State Update $\mathcal{S}^{(2)} = \{m\in\mathcal{G} \big| R_m(P_m^{(1)})>0\}$;
      \State Run Algorithm~\ref{algo_optimal_simple} with $\mathcal{S} = \mathcal{S}^{(2)}$;
      \State Update solution $\{P_m^{(2)}\}_{m\in \mathcal{S}^{(2)}}$.
    \end{algorithmic}
    \columnbreak
    \begin{algorithmic}
      \scriptsize
      \State $DP1$:
      \State Set $\mathcal{S}^{(1)} = \{1,\ldots,M\}$, $\textbf{P}^{(1)}= \textbf{P}^{(2)} = \textbf{0}$;
      \State Numerically solve $DP2$; 
      \State Update solution $\{\textbf{p}_m^{(1)}\}_{m\in \mathcal{S}^{(1)}}$;
      \State Update $\mathcal{S}^{(2)} = \{m \big| R_m(\textbf{p}_m^{(1)})>0\}$;
      \State Numerically solve $DP2$; 
      \State Update solution $\{\textbf{p}_m^{(2)}\}_{m\in \mathcal{S}^{(2)}}$.
    \end{algorithmic}
  \end{multicols}
\end{algorithm}

Here, we initially choose the Good Set as the starter in the special case according to Theorem~\ref{thm_noc1}. Since Algorithm~\ref{algo_optimal_simple} only solves the modified problem $DP2^{(s)}$ (or $DP2$ for the general case) but not the original problem, the solution $\{P_m^{(1)}\}_{m\in \mathcal{S}}$ (or $\{\textbf{p}_m^{(1)}\}_{m\in \mathcal{S}}$ for the general case) in the first iteration may allocate positive power to channel states under which the rate is zero, due to the $\max(0,\cdot)$ function in the original objective function. Therefore, it is natural to exclude the channel states with zero rate and do another iteration with the remaining channel states. This is reflected in the $3rd$ and the $4th$ line in Algorithm~\ref{algo_water_filling}. We also point out that after the second iteration, every channel state will be associated with a positive rate. As shown in Theorem~\ref{thm_opt_large_p}, this algorithm is able to solve $DP1^{(s)}$ optimally when the power constraint is large enough, but in general we will show later with numerical examples, this two-step algorithm is not guaranteed to find the optimal active set. This is because it fails to capture enough characteristics of the objective function as the $\max(0,\cdot)$ function is omitted. To this end, our next algorithm (Algorithm~\ref{algo_iter_order}) attempts to take the channel state information into account. Specifically, we will include channel states in the active set by constructing a certain order on different channel states.

\begin{algorithm}[h]
\caption{Proposed iterative algorithm}\label{algo_iter_order}
  \begin{multicols}{2}
    \begin{algorithmic}[1]
      \scriptsize
      \Statex $DP1^{(s)}$:
      \State Set $\mathcal{S} = \mathcal{G}_s$, where $\mathcal{G}_s$ is the ordered version of the good set in the ascending order according to the ordering method, i.e., from the bad channels to the good ones;
      \State Set $\{P_m\}_{m=1}^M = \{P_m^\prime\}_{m=1}^M = \textbf{0}$;
      \State Set $\mathcal{S}^\prime = \mathcal{S}$;
      \State Let $r(P_m) = \sum_{m\in\mathcal{S}} f_m R_m^+(P_m)$;
      \While{$r(P_m^\prime) \geq r(P_m)\;\&\&\; \mathcal{S}^\prime \neq \emptyset$}
      \State $\{P_m\}_{m=1}^M \gets \{P_m^\prime\}_{m=1}^M$, $\{P_m^\prime\}_{m=1}^M \gets \textbf{0}$, $\mathcal{S} \gets \mathcal{S}^\prime$;
      \State Run Algorithm~\ref{algo_optimal_simple} $\;\rightarrow\;$ Update solution $\{P_m^\prime\}_{m\in\mathcal{S}}$;
      \If{$\exists m$ such that $R_m^+(P_m^\prime) = 0$}
      \State $\mathcal{\tilde S} = \{m\in\mathcal{S} \big| R_m(P_m^\prime)>0\}$;
      \State $\{P_m^\prime\}_{m\notin\mathcal{\tilde S}} = \textbf{0}$;
      \State Run Algorithm~\ref{algo_optimal_simple} with $\mathcal{\tilde S} \;\rightarrow\;$ Update $\{P_m^\prime\}_{m\in\mathcal{\tilde S}}$;
      \EndIf
      \State $\mathcal{S}^\prime \gets \mathcal{S}\backslash\{1$st channel state$\}$;
      \EndWhile
      \State $\{P_m^*\}_{m=1}^M = \{P_m\}_{m=1}^M$.
    \end{algorithmic}
    \columnbreak
    \begin{algorithmic}
      \scriptsize
      \State $DP1$:
      \State Set $\mathcal{S}$ as the ordered version of $\{1,\ldots,M\}$ in the ascending order according to the ordering method, i.e., from the bad channels to the good ones;
      \State Set $\textbf{P} = \textbf{P}^\prime = \textbf{0}$, where $\textbf{P} = \{\textbf{p}_1,\ldots,\textbf{p}_M\}$;
      \State Set $\mathcal{S}^\prime = \mathcal{S}$;
      \State Let $r(\textbf{P}) = \sum_{m\in\mathcal{S}} f_m R_m^+(\textbf{p}_m)$;
      \While{$r(\textbf{P}^\prime) \geq r(\textbf{P})\;\&\&\; \mathcal{S}^\prime \neq \emptyset$}
      \State $\textbf{P} \gets \textbf{P}^\prime$, $\textbf{P}^\prime \gets \textbf{0}$, $\mathcal{S} \gets \mathcal{S}^\prime$;
      \State Solve $DP2$ $\;\rightarrow\;$ Update solution $\{\textbf{p}_m^\prime\}_{m\in\mathcal{S}}$;
      \If{$\exists m$ such that $R_m^+(\textbf{p}_m^\prime) = 0$}
      \State $\mathcal{\tilde S} = \{m\in\mathcal{S} \big| R_m(\textbf{p}_m^\prime)>0\}$;
      \State $\{\textbf{p}_m^\prime\}_{m\notin\mathcal{\tilde S}} = \textbf{0}$;
      \State Solve $DP2$ with $\mathcal{\tilde S} \;\rightarrow\;$ Update $\{\textbf{p}_m^\prime\}_{m\in\mathcal{\tilde S}}$;
      \EndIf
      \State $\mathcal{S}^\prime \gets \mathcal{S}\backslash\{1$st channel state$\}$;
      \EndWhile
      \State $\textbf{P}^* = \textbf{P}$;.
    \end{algorithmic}
  \end{multicols}
\end{algorithm}

The idea of Algorithm~\ref{algo_iter_order} is to sort the channel states in a certain order, which indicates the goodness of the channels, and solve the simplified problem $DP2^{(s)}$ (or $DP2$ in the general case) with the ordered channel set shrunk by one channel state per iteration. Specifically, for each iteration, we successively eliminate the worst channel in the set $\mathcal{S}$, until the sum rate starts to decrease. Therefore, the performance of this algorithm highly depends on the choice of the ordering methods. In fact, the water-filling algorithm for the symmetric power policy case also has an implicit order of the channel states. The KKT solution shows the power allocation grows monotonically with $1/\lambda$. Since $\lambda$ is a common parameter for all channel states, the channels will become active (with positive power allocated) gradually in a certain order as $1/\lambda$ increases. Thus, the order criterion of the water-filling algorithm is just the order in which channels become active as $1/\lambda$ increases. This particular ordering does give an optimal solution in some cases but is in general suboptimal, as we will show in later examples. To this end, we propose two other ordering method with explicit characterizations.

\textbf{Ordering method 1:} The first ordering method is to order the channel states based on the point where $R_m(P_m) = 0$ in the symmetric power policy case, i.e., $P_m = \frac{||\textbf{a}||^2 - 1}{||\textbf{h}_m||^2-(||\textbf{h}_m||^2||\textbf{a}||^2-(\textbf{h}_m^T\textbf{a})^2)}$. So we can order the channel states by the following criterion
\begin{equation}\label{eq_c1}
    O^{(1)}_m = ||\textbf{h}_m||^2-(||\textbf{h}_m||^2||\textbf{a}||^2-(\textbf{h}_m^T\textbf{a})^2),
\end{equation}
which indicates the goodness of the channel states when the power constraint is small. A larger $O^{(1)}_m$ means the rate of the channel state $m$ needs a smaller power to be positive, which is better.

\textbf{Ordering method 2:} The second ordering method is to order the channel states based on the asymptotic value of $R_m(P_m)$ in the symmetric power policy case, i.e., $\lim_{P_m\rightarrow\infty}R_m(P_m) = \frac{1}{2}\log\frac{||\textbf{h}_m||^2}{||\textbf{h}_m||^2||\textbf{a}||^2-(\textbf{h}_m^T\textbf{a})^2}$. So we can order the channel states by the following criterion
\begin{equation}\label{eq_crtr2}
    O^{(2)}_m = \frac{||\textbf{h}_m||^2}{||\textbf{h}_m||^2||\textbf{a}||^2-(\textbf{h}_m^T\textbf{a})^2}.
\end{equation}
It shows the goodness of the channel states when the power constraint is large. With a larger $O^{(2)}_m$, the channel state $m$ has a larger rate asymptotically.

There is a clear intuition behind the above two ordering methods. The norm of the channel coefficient $||\textbf{h}_m||^2$ indicates the strength of the channel where a larger norm corresponds to a better channel. The difference term $||\textbf{h}_m||^2||\textbf{a}||^2-(\textbf{h}_m^T\textbf{a})^2$ indicates how the coefficients $\textbf{a}$ is "aligned" with the channel coefficients. Notice this term is always non-negative and is zero if and only if when $\textbf{h}_m$ and $\textbf{a}$ are collinear. Hence a larger difference corresponds to a worse channel. There is a tension between these two terms, and the two ordering methods captures this tension in either additively (method 1) or multiplicatively (method 2).

Apart from the algorithms mentioned before one can exhaustively evaluate all possible active sets, and solve the corresponding power allocation policies. Notice that this algorithm can guarantee the optimality for the symmetric power allocation problem $DP1^{(s)}$, while it can only guarantee a local minimum in the asymmetric power policy case. This is because even if an optimal active set is given, the corresponding $DP2$ is still non-convex, and a numerical solution is not guaranteed to a global minimum. The procedure is presented in Algorithm~\ref{algo_exhaustive_search}.

\begin{algorithm}[h]
\caption{Exhaustive search}\label{algo_exhaustive_search}
  \begin{multicols}{2}
    \begin{algorithmic}[1]
      \scriptsize
      \Statex $DP1^{(s)}$:
      \State Identify all the possible active set $\mathcal{A}_1,\ldots,\mathcal{A}_{2^M}$;
      \For{$i = 1:2^M$}
      \State Set $\{P_m^{(i)}\}_{m=1}^M = \textbf{0}$;
      \State $\mathcal{S} = \mathcal{A}_i$;
      \State Run Algorithm~\ref{algo_optimal_simple} $\;\rightarrow\;$ Update solution $\{P_m^{(i)}\}_{m\in\mathcal{S}}$;
      \EndFor
      \State $\{P_m^*\}_{m=1}^M = \max_{\{P_m^{(i)}\}_{m=1}^M} \sum_{m\in\mathcal{A}_i} f_m R_m^+(P_m^{(i)})$.
    \end{algorithmic}
    \columnbreak
    \begin{algorithmic}
      \scriptsize
      \State $DP1$:
      \State Identify all the possible active set $\mathcal{A}_1,\ldots,\mathcal{A}_{2^M}$;
      \For{$i = 1:2^M$}
      \State Set $\textbf{P}^{(i)} = \textbf{0}$;
      \State $\mathcal{S} = \mathcal{A}_i$;
      \State Numerically solve $DP2$ $\;\rightarrow\;$ Update solution $\{\textbf{p}_m^{(i)}\}_{m\in\mathcal{S}}$;
      \EndFor
      \State $\textbf{P}^* = \max_{\textbf{P}^{(i)}} \sum_{m\in\mathcal{A}_i} f_m R_m^+(\textbf{p}_m^{(i)})$.
    \end{algorithmic}
  \end{multicols}
\end{algorithm} 

The exhaustive search is not applicable for large $M$ due to its computational complexity. As an informal argument, let us denote the computational complexity of solving $DP2^{(s)}$ (or $DP2$ in the general case) by $c$. Then the complexity of Algorithm~\ref{algo_exhaustive_search} is given by $2^M c$, which increases exponentially with $M$. For other algorithms, \ref{algo_const_power} essentially requires no computation, while \ref{algo_water_filling} has a complexity of $2c$. The complexity of Algorithm~\ref{algo_iter_order} depends on the number of the iterations it runs and the number of the ordering methods it applies. For each method, the complexity is at most $2Mc$, which increases linearly with $M$. Note that if the power constraint is large enough as shown in Theorem~\ref{thm_opt_large_p}, the computational complexity is only $c$ for the optimal solution in the symmetric policy case.

\subsection{Numerical results for the discrete case}

\subsubsection{\textbf{Example 1}}
We look at a two-user example ($L=2$). Assume the channel states are distributed independently in the following Table~\ref{tab_csi_1}, and the computation coefficients are given by $\textbf{a} = (1\;1)^T$. Then the joint channel state distribution is given by Table~\ref{tab_joint_csi_1}. There are totally four possible channel states ($M=4$). According to Definition~\ref{def_good_set}, the good set in this case is given by $\mathcal{G} = \{1,2,3,4\}$.

\begin{table}[!ht]
    \centering
    \caption{Channel state distribution in Example 1}
    \label{tab_csi_1}
    \begin{tabular}{|c|c|c||c|c|c|}
        \hline
        $h_1$ & $1$ & $3$ & $h_2$ & $0.5$ & $2$ \\
        \hline
        $p_1$ & $0.6$ & $0.4$ & $p_2$ & $0.8$ & $0.2$ \\
        \hline
    \end{tabular}
\end{table}

\begin{table}[!ht]
    \centering
    \caption{Joint channel state distribution in Example 1}
    \label{tab_joint_csi_1}
    \begin{tabular}{|c|c|c|c|c|}
        \hline
        $\textbf{h}_m^T$ & $(1,0.5)$ & $(1,2)$ & $(3,0.5)$ & $(3,2)$ \\
        \hline
        $f_m$ & $0.48$ & $0.12$ & $0.32$ & $0.08$ \\
        \hline
    \end{tabular}
\end{table}

Let us first consider the symmetric power allocation policy case. We will first derive $\bar P_o$ from Theorem \ref{thm_opt_large_p}. We solve the equation $R_m^\prime(x_m)x_m = R_m(x_m)$ for $x_m$ for $m=1,2,3,4$, and then solve $\lambda_m$ from $P_m^{(KKT)} (\lambda_m) = x_m$. Then we take $\lambda_o = \min \lambda_m$. Thus, $\bar P_o$ is finally given by $\sum_{m=1}^4 f_m P_m^{(KKT)} (\lambda_o) = 2.09$.  

\begin{figure}[!ht]
    \centering
    \includegraphics[width = 0.7\linewidth]{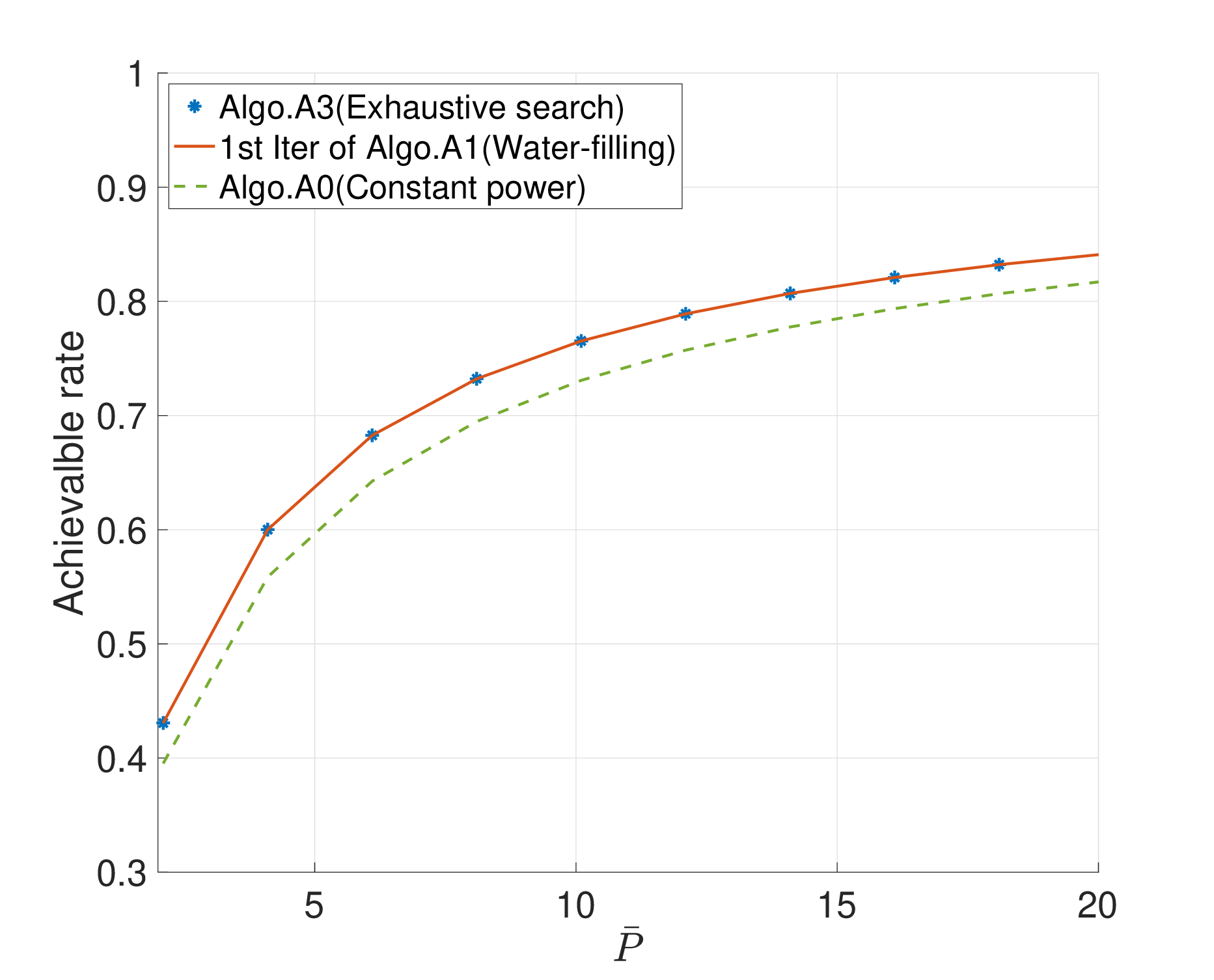}
    \caption{\small Achievable rate in the symmetric power policy case with $\bar P > \bar P_o$ in Example 1. We observe that Algorithm~\ref{algo_water_filling} gives the same result as Algorithm~\ref{algo_exhaustive_search} (hence optimal), as proven in Theorem~\ref{thm_opt_large_p}.}
    \label{fig_power_alloc_3_largeP}
\end{figure}

Fig.~\ref{fig_power_alloc_3_largeP} plots the expected achievable rate using the exhaustive search, the first iteration of Algorithm~\ref{algo_water_filling} and constant power policy with the power constraint larger than $\bar P_o = 2.09$. As discussed in the previous section, Algorithm~\ref{algo_exhaustive_search} (exhaustive search) will achieve the optimal expected rate. It can be observed that the first iteration of Algorithm~\ref{algo_water_filling} ($\{P_m^{(1)}\}_{m=1}^M$) also achieves optimality, which is consistent with Theorem~\ref{thm_opt_large_p}. Moreover, the average rate of Algorithm~\ref{algo_const_power} has a constant gap below the optimality for all power constraints.





\begin{figure}[!ht]
    \centering
    \includegraphics[width = 0.7\linewidth]{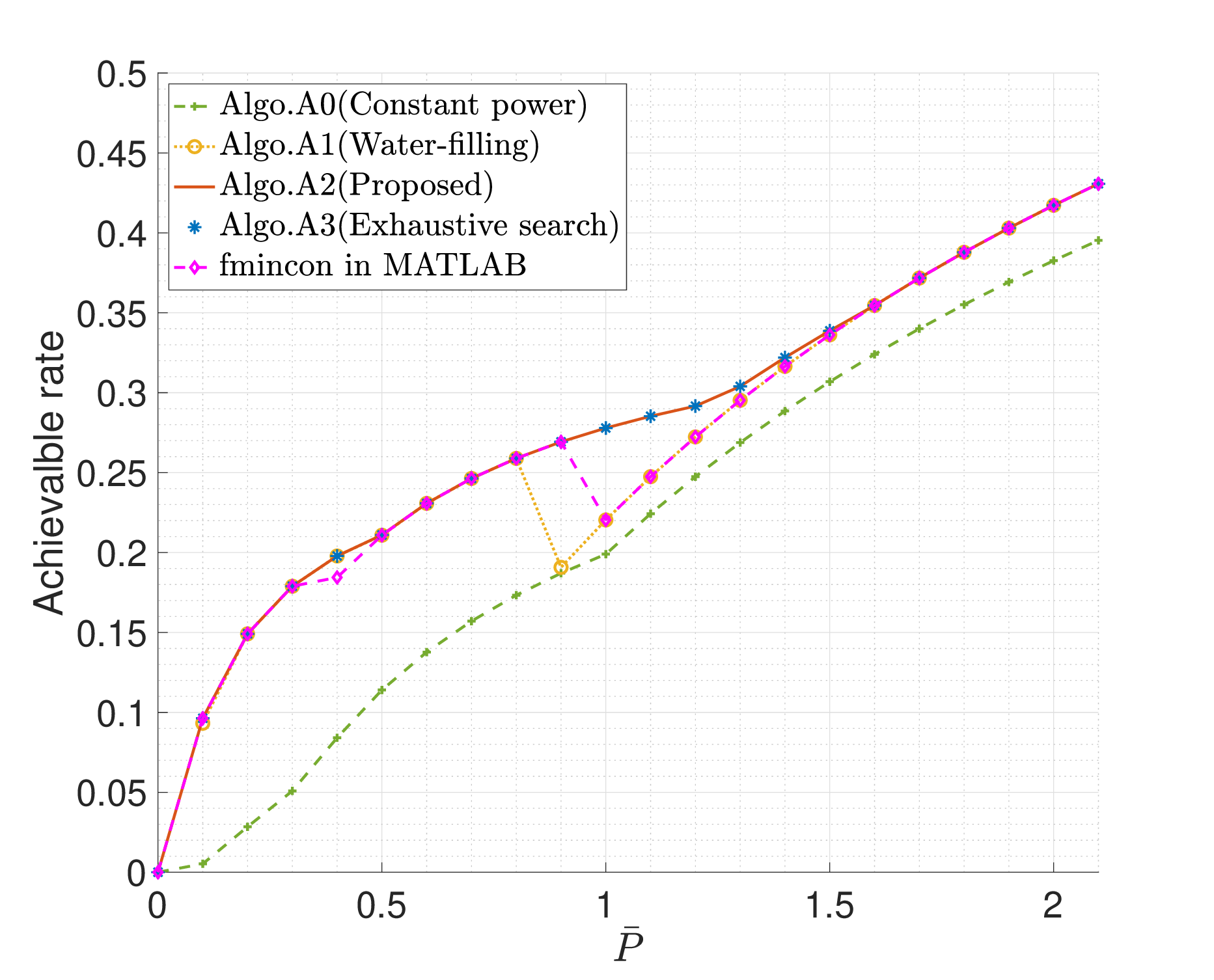}
    \caption{\small Achievable rate in the symmetric power policy case with $\bar P \leq \bar P_o$ in Example 1. We observe that Algorithm~\ref{algo_iter_order} gives the same result as Algorithm~\ref{algo_exhaustive_search} hence is optimal in this case.}
    \label{fig_power_alloc_3_smallP}
\end{figure}

Fig.~\ref{fig_power_alloc_3_smallP} shows the expected achievable rate for the power constraint no larger than $\bar P_o$. Algorithm~\ref{algo_exhaustive_search} (exhaustive search) gives the optimal expected rate. It is clear that Algorithm~\ref{algo_const_power} has the poorest performance with the trivial power allocation policy. Algorithm~\ref{algo_water_filling} is optimal for some power constraints but is in general sub-optimal. 
A sharp rate falling occurs from $\bar P = 0.8$ to $0.9$ because the power allocation policy derived by Algorithm~\ref{algo_water_filling} is regarding the modified objective function without $(\cdot)^+$. However, the power policy is applied to the primal objective when calculating the expected rate. 
Finally, we run Algorithm~\ref{algo_iter_order} with both ordering methods. For any value of the power constraint (up to $\bar P_o$), at least one of the two ordering methods finds the optimal power allocation policy. By taking the larger rate of the two, Fig.~\ref{fig_power_alloc_3_smallP} shows that our proposed algorithm is optimal for this example for all values of power constraints. We also run an off-the-shelf algorithm (\texttt{fmincon} in MATLAB) and find that it has a similar performance to the water-filling algorithm (A1) which is suboptimal for some power constraints.

Fig.~\ref{fig_power_alloc_3_APSR} shows an example for the general (asymmetric power allocation policy) case where the channel statistics are the same with the symmetric power policy example. We apply Algorithm~\ref{algo_water_filling}, \ref{algo_iter_order} and \ref{algo_exhaustive_search} for $DP1$, where both ordering methods are applied and the larger rate is taken. In the figure, we also plot the optimal result of the symmetric power policy case and the constant power case (Algorithm~\ref{algo_const_power}) for comparison. It can be observed that Algorithm~\ref{algo_water_filling} has a similar behaviour in asymmetric case compared to the symmetric case. It can match the performance of the exhaustive algorithm at some power constraint points while falls quite far at other points. However, it still results in a larger computation rate than the optimal symmetric policy. There are still sharp rate decreases in the curve of Algorithm~\ref{algo_water_filling} due to the same reason in Fig.~\ref{fig_power_alloc_3_smallP}. For the proposed algorithm (\ref{algo_iter_order}), it matches the rate derived from the exhaustive algorithm in all the power constraint regime in this example, which implies that the ordering methods may still work well for the general power allocation policy. We also observed that in the small power constraint regime, ordering method $2$ is usually not optimal while in a larger power constraint regime, ordering method 1 will become suboptimal, which is consistent with our analysis in Sec.III. As the power constraint increases, both ordering methods achieve the optimal rate since all channel states become active. The off-the-shelf algorithm (\texttt{fmincon} in MATLAB) is also sub-optimal for some power constraints. In particular, for some small power constraints (for example, $\bar P = 0,2,0.4$), it has less rate than Algorithm~A1 (water-filling algorithm). 

\begin{figure}[!ht]
    \centering
    \includegraphics[width = 0.7\linewidth]{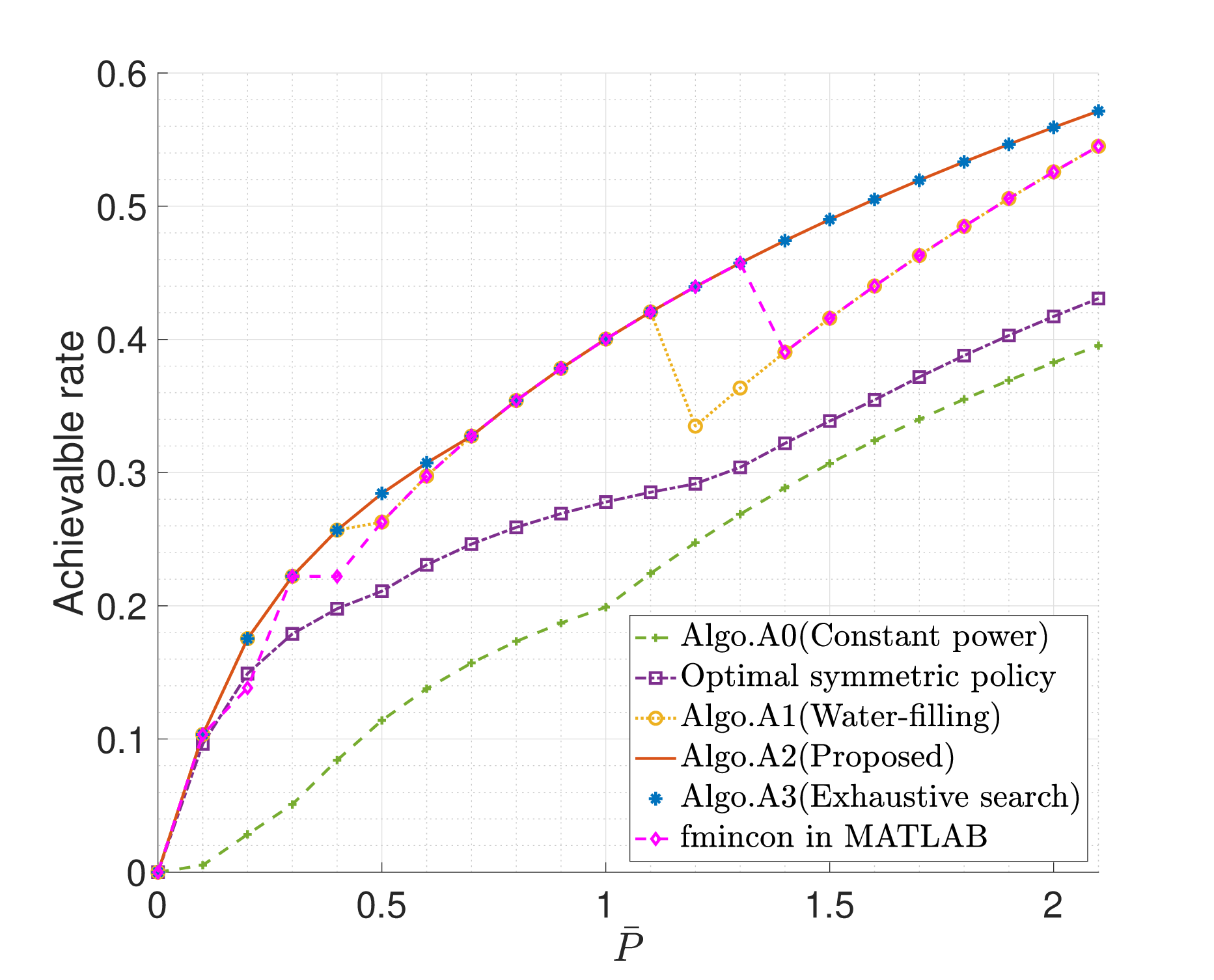}
    \caption{\small Achievable rate in the asymmetric power policy case in Example 1. We observe that Algorithm~\ref{algo_iter_order} gives the same result as Algorithm~\ref{algo_exhaustive_search} hence is optimal in this case.}
    \label{fig_power_alloc_3_APSR}
\end{figure}

\subsubsection{\textbf{Example 2}}
In this example, we will consider a two-user case, where each user has three possible channel states and their channel statistics are the same. The channel coefficients are given by $0.5,1,2.5$. The probabilities are are taken from the CDF of the Rayleigh distribution with unity scale parameter. The Probability mass function (PMF) of the channel coefficients for each user is given in Table.~\ref{tab_csi_2}. In this case, $M=9$. The computation coefficients are still given by $\textbf{a} = (1\;1)^T$. In this example, $\bar P_o$ in Theorem~\ref{thm_opt_large_p} is given by $5.02$. So we consider the power constraint regime $[0,5]$. We run the power allocation algorithms $A0-A3$ for both symmetric and asymmetric power policy scenarios and plot the average achievable rate against the power constraint in Fig~\ref{fig_power_alloc_3x3}. We only plot the results for $\bar P = [0.5,3.5]$ because the average rates for Algorithms $A0-A3$ are the same out of this range.
\begin{table}[!ht]
    \centering
    \caption{Channel state distribution in Example 2}
    \label{tab_csi_2}
    \begin{tabular}{|c|c|c|c|}
        \hline
        $h_{1,2}$ & $0.5$ & $1$ & $2.5$ \\
        \hline
        $p_{1,2}$ & $0.1175$ & $0.2760$ & 0.6065 \\
        \hline
    \end{tabular}
\end{table}

In Fig.~\ref{fig_power_alloc_3x3}, when the power constraint is very small, for example $\bar P = 0.5$, the number of the active channels is expected to be small. So both symmetric and asymmetric policy scenarios have a similar average rate. When $\bar P$ increases, as the number of the active channels also increases, the asymmetric policy gives a larger average rate than the symmetric policy. Regarding the performances of the different algorithms, Algorithm~\ref{algo_exhaustive_search} is the best since it explores all the possible active sets. Algorithm~\ref{algo_const_power} has a worst average rate with a constant gap below other algorithms. The differences between Algorithm~\ref{algo_exhaustive_search},~\ref{algo_water_filling} and~\ref{algo_iter_order} are not very big. In general, Algorithm~\ref{algo_water_filling} is almost always below the exhaustive search except for very large power constraints, where the optimality is guaranteed by Theorem.~\ref{thm_opt_large_p}. Algorithm~\ref{algo_iter_order} can achieve the exhaustive search performance for some power constraints. Specifically, in the symmetric policy scenario, it finds an optimal active set for most power constraints and fails when the power constraint is between $1.25$ and $2.25$. In the asymmetric policy scenario, it finds an optimal active set only when the power constraint is around $1.5$. In addition, Algorithm~\ref{algo_iter_order} gives better power policies than Algorithm~\ref{algo_water_filling} for almost all the power constraints in the symmetric policy scenario. It also provides a clear rate gain over Algorithm~\ref{algo_water_filling} when the power constraint is between $1$ and $2.5$ in the asymmetric power policy scenario.
\begin{figure}[!ht]
    \centering
    \includegraphics[width = 0.7\linewidth]{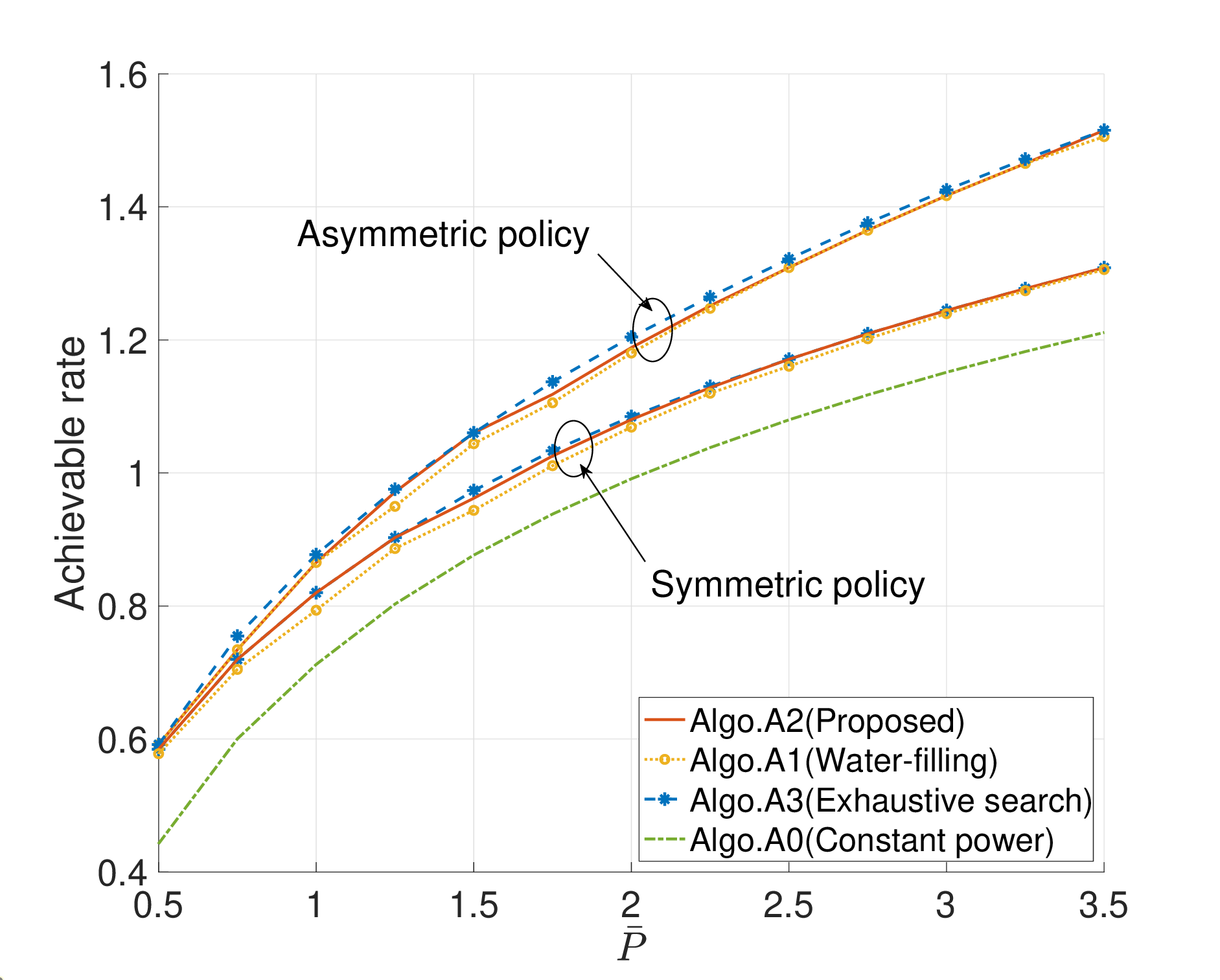}
    \caption{\small Achievable rate against the power constraint in Example 2.}
    \label{fig_power_alloc_3x3}
\end{figure}

\subsubsection{\textbf{Example 3}}
In this example, we will consider a larger number of possible channel states, which is more practically relevant. Still in the two-user case with computation coefficients given by $\textbf{a} = (1\;1)^T$, we assume each user has ten possible channel states. The channel coefficients are given by $0.25, 0.5, 0.75, 1, 1.25, 1.5, 2, 2.5, 3, 4$. The probabilities are are taken from the CDF of the Rayleigh distribution with unity scale parameter. The PMF of the channel coefficients for each user is given in Table~\ref{tab_csi_3}. In this case, $M=100$, and algorithms~\ref{algo_exhaustive_search} is not applicable due to the complexity. Note that $\bar P_o$ in Theorem~\ref{thm_opt_large_p} is given by $13.05$ in this example. So we consider the power constraint regime $[0, 13]$.We run the power allocation algorithms $A0-A2$ for both symmetric and asymmetric power policy scenarios and plot the average achievable rate against the power constraint in Fig~\ref{fig_power_alloc_6x6}. We only plot the results for $\bar P = [05,10]$ because the average rates for Algorithm~\ref{algo_water_filling} and Algorithm~\ref{algo_iter_order} are the same out of this range.
\begin{table}[!ht]
    \centering
    \caption{Channel state distribution in Example 3}
    \label{tab_csi_3}
    \begin{tabular}{|c|c|c|c|c|c|c|c|c|c|c|}
        \hline
        $h_{1,2}$ & $0.25$ & $0.5$ & $0.75$ & $1$ & $1.25$ & $1.5$ & $2$ & $2.5$ & $3$ & $4$ \\
        \hline
        $p_{1,2}$ & $0.0308$ & $0.0867$ & $0.1277$ & $0.1483$ & $0.1487$ & $0.1332$ & $0.1893$ & $0.0914$ & $0.0328$ & $0.0111$ \\
        \hline
    \end{tabular}
\end{table}

In Fig~\ref{fig_power_alloc_6x6}, the asymmetric policy gives a larger average rate than the symmetric policy even at a low power constraint, for example $\bar P = 0.5$. Through all the power constraints, Algorithm~\ref{algo_const_power} has a worst average rate with a constant gap below other algorithms. Although it looks like Algorithm~\ref{algo_iter_order} and Algorithm~\ref{algo_water_filling} have nearly the same average rates in the figure, we can find out the differences from their values. This can be clearly observed from the zoomed-in parts in Fig~\ref{fig_power_alloc_6x6}. If there is a rate difference between the algorithms, no matter how tiny it is, it indicates that different active sets are decided. Specifically, in the symmetric power policy scenario, Algorithm~\ref{algo_iter_order} gives better power policies than Algorithm~\ref{algo_water_filling} at all the considered power constraints. In the asymmetric power policy scenario, it finds better active sets than Algorithm~\ref{algo_water_filling} when the power constraint is between $4$ and $8.5$.
\begin{figure}[!ht]
    \centering
    \includegraphics[width = 0.7\linewidth]{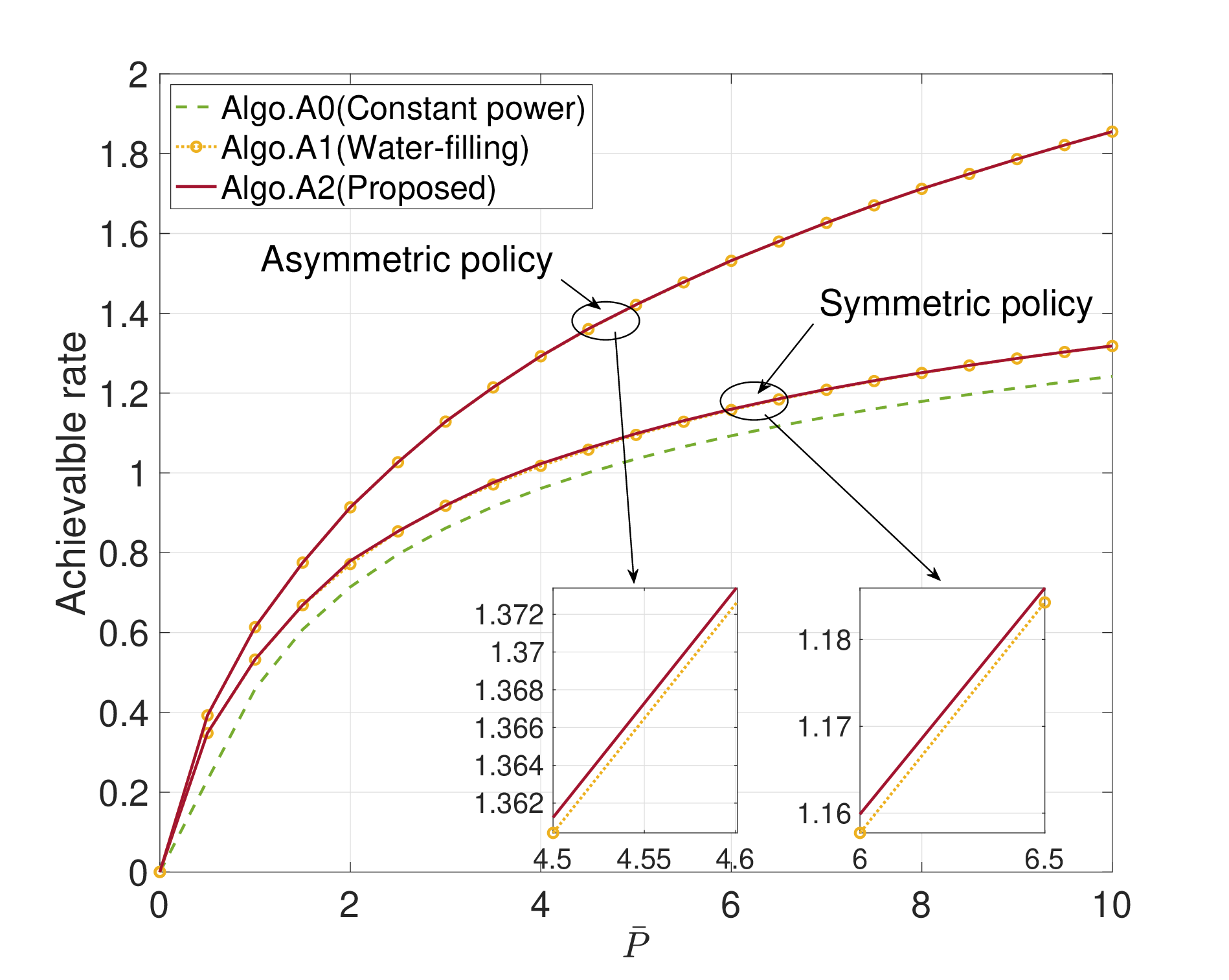}
    \caption{\small Achievable rate against the power constraint in Example 3.}
    \label{fig_power_alloc_6x6}
\end{figure}

\section{Continuous case}

Now we turn our attention to the case when the channel coefficients are modeled as continuous random variables. In this case, we do not have an active set of finite size as in the discrete case so the results from the previous section do not extend to the continuous case directly. However, by looking at the active ``support" of the channel states, similar results could be established. Also notice that in practice, channel coefficients are always quantized so that we only need to deal with a finite number of possible channel states. In this sense, the results from the discrete case are more relevant. Nevertheless, we state a few results for the symmetric power policy in the continuous case.

\subsection{Problem statement}
Assume $h_1,\ldots,h_L$ are continuously distributed and independent with each other. Then the target problem in the symmetric power policy case is given by $CP1$.
\begin{equation}\label{primal_problem_continuous}
    \begin{split}
        CP1:\quad\max_{P(\textbf{h})} \quad & \int_{\mathcal{D}(\textbf{h})} g^+(\textbf{h},P(\textbf{h})) f(\textbf{h}) d\textbf{h},\\
        s.t. \quad & \int_{\mathcal{D}(\textbf{h})} P(\textbf{h}) f(\textbf{h}) d\textbf{h} \leq \bar P,\\
        & P(\textbf{h}) \geq 0,
    \end{split}
\end{equation}
where $g(\textbf{h},x) = \frac{1}{2}\log\frac{1+||\textbf{h}||^2 x}{ ||\textbf{a}||^2 + \epsilon(\textbf{h}) x}$, $\epsilon(\textbf{h}) = ||\textbf{h}||^2||\textbf{a}||^2-(\textbf{h}^T\textbf{a})^2$, $\mathcal{D}(\textbf{h})$ is the support domain of the channel state vector and $f(\textbf{h})$ is the probability density function (PDF) of $\textbf{h}$. Similar to the discrete case, we can define the active domain $\mathcal{D}_A$ as $\mathcal{D}_A = \{ \textbf{h}\in\mathcal{D}(\textbf{h}) \big| P(\textbf{h})>0\}$, where $P(\textbf{h})$ is any feasible power policy. Furthermore, the optimal active domain $\mathcal{D}_A^*$ corresponding to any optimal power policy $P^*(\textbf{h})$ is defined as $\mathcal{D}_A^* = \{ \textbf{h}\in\mathcal{D}(\textbf{h}) \big| P^*(\textbf{h})>0\}$.
If $\mathcal{D}_A^*$ is known, straightforwardly extending Lemma~\ref{lem_simple_prob} to the continuous case, we can simplify $CP1$ to $CP2$ with its support domain $\mathcal{D}_S = \mathcal{D}_A^*$. 

\begin{equation}\label{primal_problem_convex_continuous}
    \begin{split}
        CP2:\quad\max_{P(\textbf{h})} \quad & \int_{\mathcal{D}_S} g(\textbf{h},P(\textbf{h})) f(\textbf{h}) d\textbf{h},\\
        s.t. \quad & \int_{\mathcal{D}_S} P(\textbf{h}) f(\textbf{h}) d\textbf{h} \leq \bar P,\\
        & P(\textbf{h}) \geq 0.
    \end{split}
\end{equation}

Note that $CP2$ is a convex problem, and its optimal solution can also be given by KKT conditions, which can be proven in a similar way as Lemma~\ref{lem_optimal_solution_simple}. Specifically, let us define 
\begin{equation}\label{eq_p_star_ct}
    P(\textbf{h},\mu) = \Bigg\{\begin{array}{ll}
        \left(\frac{1}{\mu}-\frac{1}{||\textbf{h}||^2}\right)^+, & \textbf{h} \text{ is collinear with } \textbf{a} \\
        \left(\frac{ - b(\textbf{h}) + \sqrt{b(\textbf{h})^2 - 4 d(\textbf{h}) c(\textbf{h},\mu)}}{2d(\textbf{h})}\right)^+, & \textbf{h} \text{ is not collinear with } \textbf{a}
        \end{array},
\end{equation}
where $d(\textbf{h}) = ||\textbf{h}||^2 \epsilon(\textbf{h})$, $b(\textbf{h}) = ||\textbf{h}||^2||\textbf{a}||^2 + \epsilon(\textbf{h})$, $c(\textbf{h},\mu) = ||\textbf{a}||^2 - (\textbf{h}^T\textbf{a})^2/\mu$ and $\mu>0$. Then the optimal solution to $CP2$ can be given by
\begin{equation}\label{eq_CP2_solution}
    P^\star(\textbf{h}) = P(\textbf{h},\mu^*),
\end{equation}
with
\begin{equation}\label{eq_CP2_constraint}
    \int_{\mathcal{D}_S} P(\textbf{h},\mu^*) f(\textbf{h}) d\textbf{h} = \bar P.
\end{equation}
Practically, bisection algorithm can be used to derive the optimal $\mu^*$, and then the optimal solution is given by $P(\textbf{h},\mu^*)$. This is presented in Algorithm~\ref{algo_optimal_simple_c}.

\renewcommand{\thealgorithm}{\arabic{algorithm}}
\setcounter{algorithm}{1}

\begin{algorithm}
\caption{Optimal algorithm for $CP2$}\label{algo_optimal_simple_c}
\begin{algorithmic}[1]
\scriptsize
\State Define $q(\mu) = \int_{\mathcal{D}_S} P(\textbf{h},\mu) f(\textbf{h}) d\textbf{h}$;
\State Choose a small enough $\mu_l$ such that $q(\mu_l) > \bar P$;
\State Choose a large enough $\mu_u$ such that $q(\mu_u) < \bar P$;
\State $\mu \gets \frac{\mu_l+\mu_u}{2}$;
\While{$\bar P-q(\mu)>\epsilon \;\big|\big|\; \bar P - q(\mu)<0$} \Comment{$\epsilon$ is a small number, e.g. $10^{-3}$}
\If{$\bar P-q(\mu)>0$}
    \State $\mu_u \gets \mu$;
\Else
    \State $\mu_l \gets \mu$;
\EndIf
\State $\mu \gets \frac{\mu_l+\mu_u}{2}$;
\EndWhile
\State $P^*(\textbf{h}) = P(\textbf{h},\mu)$.
\end{algorithmic}
\end{algorithm}

To study the property of the optimal active domain in the continuous case, we define the good domain $\mathcal{D}_G$ and the bad domain $\mathcal{D}_B$ as $\mathcal{D}_G = \{ \textbf{h}\in\mathcal{D}(\textbf{h}) \big| ||\textbf{h}||^2 > \epsilon(\textbf{h})\}$ and $\mathcal{D}_B = \{ \textbf{h}\in\mathcal{D}(\textbf{h}) \big| ||\textbf{h}||^2 \leq \epsilon(\textbf{h})\}$ respectively. Using the same argument as in the discrete case, it is easy to see that the optimal domain does not contain $\mathcal{D}_B$, that is $\mathcal{D}_A^* \cap \mathcal{D}_B = \emptyset$. It is tempting to conjecture that similar to the discrete case (Theorem~\ref{thm_opt_large_p}), there exists a threshold on the power constraint, above which the optimal domain $\mathcal{D}_A^*$ is equal the good domain. However, we now show that this is not always the case in the continuous case.

Let us assume $\mathcal{D}_A^* = \mathcal{D}_G$. Then the optimal solution can be given by Algorithm~\ref{algo_optimal_simple_c} and its analytical form is given by $P(\textbf{h},\mu^*)$ in (\ref{eq_p_star_ct}). Using the same argument as in Lemma~\ref{lem_optimal_active_set} for the continuous case, it follows that
\begin{equation}\label{eq_CP2_P_ineq}
    P(\textbf{h},\mu^*) > \frac{||\textbf{a}||^2 -1}{||\textbf{h}||^2-\epsilon(\textbf{h})},\forall \textbf{h} \in \mathcal{D}_G,
\end{equation}
Substituting (\ref{eq_CP2_solution}) into (\ref{eq_CP2_P_ineq}), it holds that
\begin{equation}\label{ineq_mu}
    \frac{1}{\mu^*} > \Bigg\{\begin{array}{ll}
    \frac{||\textbf{a}||^2}{||\textbf{h}||^2}, & \epsilon(\textbf{h}) = 0, \\
    \frac{(\textbf{h}^T\textbf{a})^2}{||\textbf{h}||^2 - \epsilon(\textbf{h})}, & \epsilon(\textbf{h}) \neq 0.
    \end{array}
\end{equation}
This means when $\mathcal{D}_A^* = \mathcal{D}_G$, the above inequality holds for all $\textbf{h} \in \mathcal{D}_G$. However, from the definition of the good domain $\mathcal{D}_G = \{ \textbf{h}\in\mathcal{D}(\textbf{h}) \big| ||\textbf{h}||^2 > \epsilon(\textbf{h})\}$, the inequality in (\ref{ineq_mu}) implies that $1/\mu^*$ cannot be bounded from above if there exists a sequence of $\textbf{h}\in\mathcal{D}_G$ such that $||\textbf{h}||^2$ can approach $\epsilon(\textbf{h})$ asymptotically. On the other hand, an unbounded $1/\mu^*$ implies that the power constraint $\bar P$ cannot be bounded from above as well. This is because of (\ref{eq_CP2_constraint}) and the fact that $P(\textbf{h},\mu)$ in (\ref{eq_p_star_ct}) grows at least linearly with $\sqrt{1/\mu}$. This shows that in general there does not exist an upper bound on the power constraint, above which $\mathcal{D}_A^*$ is equal to $\mathcal{D}_G$ in the continuous case. In the following, we study the power allocation algorithms for the continuous case.

\subsection{Power allocation algorithms for the continuous case}
Unlike the discrete case, we cannot apply Algorithm~\ref{algo_exhaustive_search} (exhaustive search) in the continuous case to derive the optimal solution. However, the algorithms \ref{algo_const_power} and \ref{algo_water_filling} can be straightforwardly extended to the continuous case. Let us keep their captions as in the discrete case. Moreover, the proposed ordering methods can also be extended to the continuous case with the ordering metric defined by $O^{(1)} = ||\textbf{h}||^2-(||\textbf{h}||^2||\textbf{a}||^2-(\textbf{h}^T\textbf{a})^2)$ and $O^{(2)} = \frac{||\textbf{h}||^2}{||\textbf{h}||^2||\textbf{a}||^2-(\textbf{h}^T\textbf{a})^2}$. Unlike the discrete case, the channel states in the continuous case cannot be ordered into a finitely countable set. However, the intuition from the discrete case still carries over, namely, the larger the ordering criterion is, the better the channel is. So we will search the optimal active domain in the direction depending on the ordering criterion. In each iteration, the active domain will be shaped further along the direction where the ordering criterion is larger. This intuition is substantiated in Algorithm~\ref{algo_iter_order_c}.

\renewcommand{\thealgorithm}{A\arabic{algorithm}}
\setcounter{algorithm}{3}

\begin{algorithm}
\caption{Proposed iterative algorithm for $CP1$}\label{algo_iter_order_c}
\begin{algorithmic}[1]
\scriptsize
\State Define $r(P(\textbf{h})) = \int_{\mathcal{D}_S} g(\textbf{h},P(\textbf{h})) f(\textbf{h}) d\textbf{h}$;
\State $C_t = 0$, $s = 0.1$; \Comment{$C_t$ is the shaping parameter, $s$ is the step size}
\State $P(\textbf{h}) = P^\prime(\textbf{h}) = 0$;
\While{$r(P^\prime(\textbf{h})) - r(P(\textbf{h})) > \epsilon$} \Comment{$\epsilon$ is a small number, e.g. $10^{-3}$}
\State $\mathcal{D}_S \gets \{\textbf{h} \big| O(\textbf{h})>C_t\}$; \Comment{$O(\textbf{h})$ is the ordering method, it can be either $O^{(1)}$ or $O^{(2)}$ in this case}
\State $P(\textbf{h}) \gets P^\prime(\textbf{h})$;
\State Run Algorithm~\ref{algo_optimal_simple_c} $\rightarrow$ Solution $P^\prime(\textbf{h})$;
\State $C_{t} = C_{t} + s$;
\EndWhile
\State $P^*(\textbf{h}) = P(\textbf{h})$.
\end{algorithmic}
\end{algorithm}


\subsection{Numerical results for the continuous case}
We consider a two-user case where $L=2$ and $\textbf{a} = (1\; 1)^T$. Assume the channel state vector follow the Gaussian distribution with zero mean and unit variance. Since we assume non-negative real channel coefficients, the PDF is given by
\begin{equation}
    f(h_1,h_2) = \frac{2}{\pi} e^{-\frac{h_1^2+h_2^2}{2}}.
\end{equation}
Fig.~\ref{fig_power_alloc_ct} shows the achievable computation rate under Algorithm~\ref{algo_const_power}, \ref{algo_water_filling} and \ref{algo_iter_order_c} with $O^{(1)}$. It can be observed that Algorithm~\ref{algo_const_power} gives the worst expected rate, while the proposed iterative algorithm~\ref{algo_iter_order_c} has the best rate. However, it only outperforms Algorithm~\ref{algo_water_filling} with a small gain.

\begin{figure}[t]
    \centering
    \includegraphics[width = 0.7\linewidth]{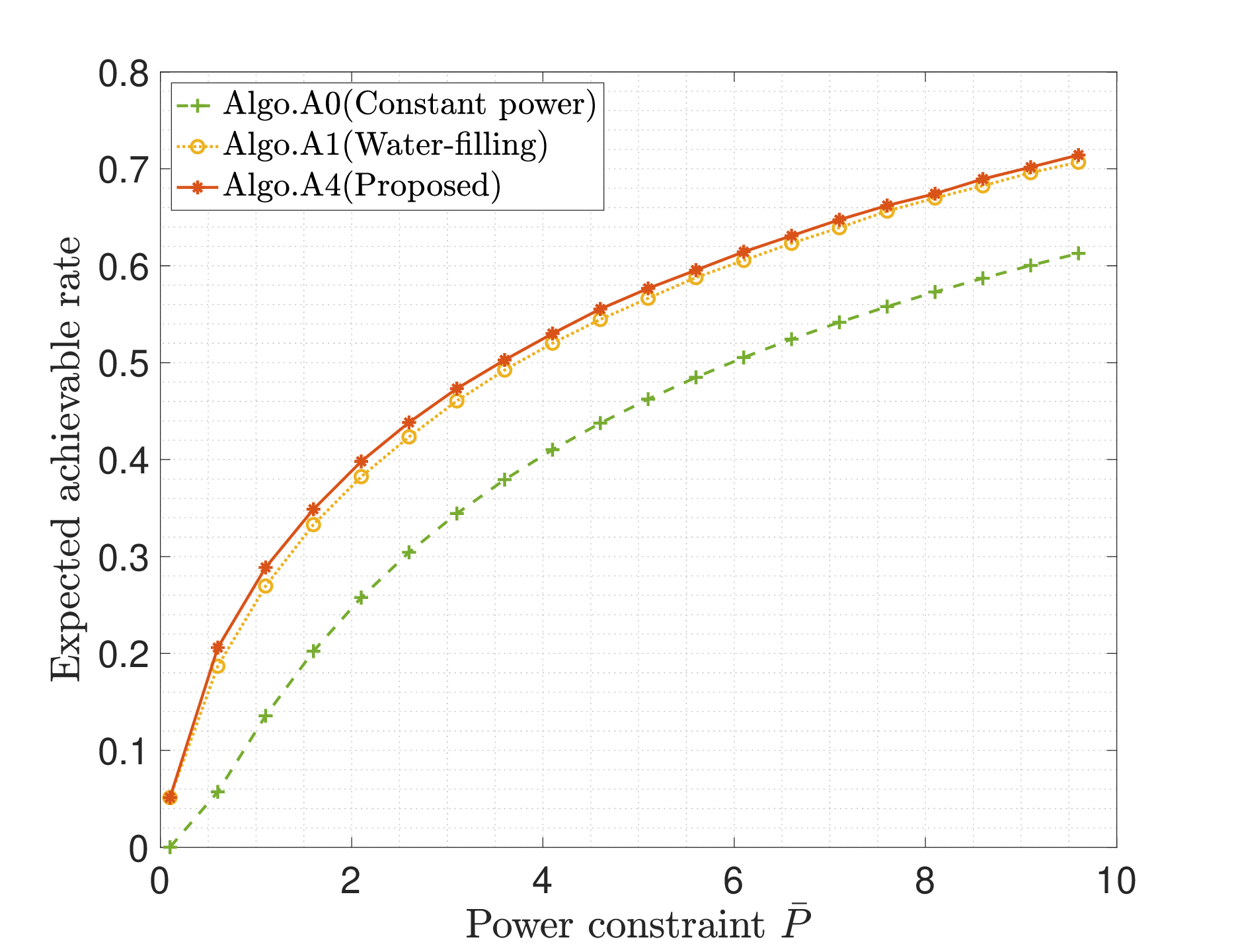}
    \caption{Performance comparison for different algorithms in the continuous case}
    \label{fig_power_alloc_ct}
\end{figure}

\section{Conclusion}

The power allocation problem for CF over fading channels is nontrivial due to its non-convexity and the $\log^+(\cdot)$ function in the rate expression. We investigate the properties of the optimal solutions. We notice that the classical water-filling like algorithm provides a sub-optimal solution with an implicit ordering methods for the channel states. To improve that, we propose two explicit ordering methods and an order-based iterative algorithm to determine the active set and find out the corresponding power policy. The proposed iterative algorithm gives a better rate than the classical algorithm while remains a linear complexity. It achieves optimality in certain examples, but in general gives good but suboptimal solutions. This possibly implies that the two ordering methods are not able to capture all the properties of the channel states. One way to improve the proposed algorithm is to find out a more judicious ordering method and feed it into the iterative algorithm frame. A potential ordering method is to combine the two proposed methods with different weights. In this case, the future research may focus on how to decide the weights. Moreover, the power constraint could be included in the ordering methods as well. It is possible that at different power constraints the ordering of the channel states are different.


\bibliographystyle{IEEEtran}
\bibliography{reference,IEEEabrv}

\end{document}